\documentclass[a4paper,11pt]{article}

\makeatletter
\@namedef{ver@etex.sty}{2016/08/01 v2.7 eTeX basic definition package}
\makeatother

\usepackage[utf8]{inputenc}
\usepackage[T1]{fontenc}
\usepackage[margin=1in]{geometry}
\usepackage{microtype}
\usepackage{natbib}

\usepackage{amsthm}

\theoremstyle{plain}
\newtheorem{theorem}{Theorem}
\newtheorem{proposition}{Proposition}
\newtheorem{lemma}{Lemma}

\newtheorem{observation}{Observation}

\theoremstyle{definition}

\newtheorem{definition}{Definition}

\usepackage[dvipsnames]{xcolor}

\definecolor{LinkBlue}{HTML}{1F4E79}
\definecolor{CiteBlue}{HTML}{2F6FA3}
\definecolor{URLBrick}{HTML}{8A3B12}
\usepackage[colorlinks=true]{hyperref}
\hypersetup{
    linkcolor=LinkBlue,
    citecolor=CiteBlue,
    urlcolor=URLBrick,
    filecolor=LinkBlue,
    pdfborder={0 0 0},
    breaklinks=true,
    bookmarksopen=true,
    bookmarksnumbered=true,
    hypertexnames=false,
    pdftitle={Efficiently Restructuring Sovereign Debt via Arctic Auctions with Convex Costs},
    pdfauthor={Jugal Garg, Edwin Lock, Vijay V. Vazirani},
    pdfkeywords={Arctic auction, supply costs, competitive equilibrium, primal-dual algorithms, sovereign debt restructuring},
}

\usepackage{amsmath, amssymb, mathtools}
\usepackage{enumerate}
\usepackage{bm}
\usepackage[capitalize]{cleveref}
\crefname{observation}{Observation}{Observations}
\Crefname{observation}{Observation}{Observations}

\usepackage{tikz}
\usetikzlibrary{shapes,patterns,calc,spy}
\usepackage{tikz-network}
\usepackage{pgfplots}
\pgfplotsset{compat=1.18}
\usetikzlibrary{arrows,automata,math,backgrounds,decorations.pathreplacing,decorations.markings,patterns,arrows.meta,positioning}
\tikzset{agent/.style={draw, circle, inner sep=5, outer sep=5}}
\tikzset{trade/.style={draw, thick, -Latex}}

\usepackage{algorithm}
\usepackage{algpseudocodex}

\makeatletter
\newcommand{\newalgblock}[4][0]{%
  \ifnum#1=0
    \algdef{SE}[#2]{#3}{End#3}{\algpx@startCodeCommand\algpx@startIndent#4}{}%
  \else\ifnum#1=1
    \algdef{SE}[#2]{#3}{End#3}[1]{\algpx@startCodeCommand\algpx@startIndent#4}{}%
  \else
    \algdef{SE}[#2]{#3}{End#3}[2]{\algpx@startCodeCommand\algpx@startIndent#4}{}%
  \fi\fi
  \expandafter\pretocmd\csname #3\endcsname{\algpx@endCodeCommand}{}{}%
  \expandafter\pretocmd\csname End#3\endcsname{\algpx@endCodeCommand[0]}{}{}%
  \algtext*{End#3}%
}
\makeatother

\newalgblock{OuterLoop}{OuterLoop}{\textbf{Outer Loop}}
\newalgblock{InnerLoop}{InnerLoop}{\textbf{Inner Loop}}
\newalgblock[2]{Event}{Event}{\textbf{Event #1:} #2}

\renewcommand{\vec}[1]{\bm{#1}}
\newcommand{\R}{\mathbb{R}}
\newcommand{\Q}{\mathbb{Q}}
\DeclareMathOperator{\conv}{\textnormal{conv}}
\DeclareMathOperator*{\argmax}{\textnormal{arg\,max}}
\newcommand{\eb}{\vec{e}}
\newcommand{\pb}{\vec{p}}
\newcommand{\qb}{\vec{q}}
\newcommand{\ub}{\vec{u}}
\newcommand{\xb}{\vec{x}}
\newcommand{\yb}{\vec{y}}
\newcommand{\zb}{\vec{z}}
\newcommand{\bb}{\vec{b}}
\newcommand{\mb}{\vec{m}}
\newcommand{\tb}{\vec{t}}
\newcommand{\db}{\vec{d}}
\newcommand{\capac}{\mathrm{cap}}
\newcommand{\seg}{\mathrm{seg}}

\allowdisplaybreaks

\title{Efficiently Restructuring Sovereign Debt via Arctic Auctions with Convex Costs}
\author{Jugal Garg\thanks{University of Illinois Urbana-Champaign. Supported by NSF Grant CCF-2334461. \texttt{jugal@illinois.edu}} \and Edwin Lock\thanks{King's College London. \texttt{edwin.lock@kcl.ac.uk}} \and Vijay V. Vazirani\thanks{University of California, Irvine. \texttt{vazirani@ics.uci.edu}}}
\date{}

\begin{document}

\maketitle

\begin{abstract}
We study the problem of computing competitive equilibria in the Arctic product-mix auction, originally developed for the Icelandic government for exchanging blocked financial accounts, and more recently proposed by IMF staff for sovereign debt restructuring. From the buyers' perspective, the Arctic auction is equivalent to the quasi-linear Fisher market. However, unlike the standard Fisher model, the seller can express rich supply preferences through explicit supply-side costs and constraints. Despite extensive algorithmic literature on Fisher markets, the seller side has not received much attention, and no polynomial-time algorithm was previously known for computing competitive equilibrium when sellers face nontrivial costs.

We examine the natural and expressive regime of separable, stepwise-increasing marginal costs that underlie the above-stated applications. Using polyhedral theory techniques, we first show that rational inputs lead to rational-valued competitive equilibria. Motivated by this result, we develop the first polynomial-time algorithm for this setting based on a non-trivial extension of classic primal-dual balanced-flow techniques for linear Fisher markets. Our work provides a robust computational foundation for auctions with sophisticated preferences, paving the way for flexible and institutionally feasible market designs in global finance.
\end{abstract}

\section{Introduction}

The Arctic Product-Mix Auction was originally designed by Paul Klemperer for the Icelandic Government to exchange blocked accounts for other financial assets, e.g., cash or bonds.%
\footnote{For Klemperer's original description of the Arctic auction, see \citet{Klemperer2}. The Arctic auction belongs to the family of Product-Mix Auctions, originally designed for the Bank of England to provide liquidity to financial institutions following the UK financial crisis of 2008-9.}
More recently, IMF staff have proposed using the Arctic auction for the purpose of sovereign debt restructuring \citep{willems2021auction}. In this mechanism, creditors bid to exchange their claims for alternative debt instruments. Importantly, the auctioneer (sovereign), and possibly outside sellers offering additional funds, can specify supply costs which govern the quantity of goods they are willing to produce at given market prices. Both the Icelandic Government and the sovereign debt restructuring mechanism pursued the approach of computing competitive equilibria for (many) different supply schedules, and then selecting the most desirable equilibrium among these.\footnote{The precise criteria for choosing the optimal competitive equilibrium are confidential.} For this approach to be practicable, efficient algorithms for computing equilibria are necessary.

The most useful market model that underlies these applications is the Arctic auction in which the auctioneer has a stepwise increasing marginal cost curve for each good. The objective is to find equilibrium prices and allocations: each buyer should obtain a utility-maximizing bundle of goods at current prices and the auctioneer should
maximize his profit. This paper develops the first efficient algorithm for computing such a competitive equilibrium, in the general regime where one or more sellers have separable convex costs for each good.

From the buyer's perspective, the Arctic auction is equivalent to a variant of the linear Fisher market in which buyers have quasi-linear utilities: the preferences expressed by buyers in the Arctic auction can be interpreted as the aggregate demand of one or more bidders in such a Fisher market. The linear Fisher market occupies a special place in algorithmic game theory not only because of its widespread applications (see Section \ref{sec:related}), but also because of its strong algorithmic properties and its connections to convex optimization. For this market model, \citet{DPSV} gave an exact polynomial-time algorithm, henceforth referred to as {\em the DPSV algorithm}. The primal--dual balanced-flow techniques from this classic algorithm will also play a key role in our work. Equilibrium allocations and prices for this model are captured by the celebrated Eisenberg-Gale convex program [\citeyear{eisenberg}]. However, this program was not sufficient to derive an efficient algorithm for computing an \emph{exact equilibrium}---that needed an algorithm for Diophantine approximation, which was given later in \cite{Jain2007polynomial}. Later, \citet{Orlin-Fisher} provided the first strongly polynomial algorithm for computing equilibria of the linear Fisher market.

While buyer preferences have been extensively studied, the \emph{seller side} of Fisher markets remains comparatively underexplored. For Arctic-auction applications in central banking and sovereign debt restructuring, however, seller costs are essential. Separable, stepwise increasing marginal costs are both practically natural and theoretically expressive: sufficiently fine steps approximate arbitrary separable convex costs. The only prior efficient algorithmic treatment of seller costs that we are aware of handles the simpler case of constant marginal costs \citep{arctic_markets_production}. In many real-world applications, supply costs---such as soft limits on the quantity of bond types within a fiscal year---arise naturally and are an essential feature of the market design. Sellers such as governments and central banks also typically have greater institutional capacity to formulate complex, constraint-rich preferences over allocations, further motivating the study of auction designs in which sellers can express rich preferences.

Due to these applications, the Arctic auction has begun to receive algorithmic attention in recent years. \citet{Fichtl-computing} analyzed an algorithm proposed by Klemperer and DotEcon that solves the Arctic auction with stepwise increasing separable marginal costs by computing equilibria via explicit enumeration, which does not scale well with the number of goods.
\citet{baldwin2024implementing} develop reductions from Arctic auctions with stepwise increasing marginal costs to auctions without costs, and to linear Fisher markets.
\citet{arctic_markets_production} showed that equilibrium computation in this setting reduces to a \emph{rational convex program} \cite{Va.rational}, revealing deep connections to convex optimization, and obtained a polynomial-time algorithm for the Arctic auction without costs by modifying the DPSV algorithm. Most recently, \citet{garg2026stronglypolynomialalgorithmarctic} built on the work of \citet{Orlin-Fisher} to develop a strongly polynomial algorithm for the costless setting.

This paper develops the first polynomial-time algorithm for computing competitive equilibrium in the Arctic auction \emph{with seller costs}. We study the regime in which the auctioneer (or multiple sellers) have independent, convex costs for each good that is represented as a stepwise increasing marginal cost curve. Thus, beyond its theoretical guarantees, this algorithm immediately enables a richer, implementable Arctic-type auction in applications such as the IMF-proposed sovereign debt restructuring and central bank operations, where sophisticated seller preferences are both natural and institutionally feasible.

Our algorithm requires several new ideas beyond those used for the linear Fisher market, and making these ideas work together constitutes the main technical challenge of the paper. In the classical setting, the supply of each good is fixed and buyers are endowed with fixed budgets; equilibrium computation therefore proceeds by adjusting prices alone. In contrast, in our setting the supply of each good must be endogenously adjusted by the algorithm in response to prices and demand, and buyers have quasi-linear preferences, so their effective spending is not fixed in advance but depends on prices and allocations. As a result, unlike the DPSV algorithm, which updates only prices, our algorithm must jointly update prices, supplies, and buyers' remaining budgets.

The main difficulty is not only that supplies change, but that the algorithm must reason about two notions of supply. During the primal-dual process, it maintains an upper bound on the amount of each good that may be sold at the current price. However, at final prices, seller optimality also imposes lower bounds: at a marginal-cost breakpoint, the seller is indifferent over an interval of quantities, while in the interior of a segment, the seller requires the full segment quantity to be sold. Thus, the usual DPSV min-cut invariant cannot be applied.

Our solution separates the fixed-supply part of the DPSV analysis from the new seller-side bookkeeping. Within an iteration of the outer loop, the algorithm scales prices while holding actual supply fixed, so the balanced-flow and surplus-reduction arguments from the fixed-supply Arctic auction can be adapted locally. Globally, we maintain a hybrid min-cut invariant on an auxiliary network that uses actual supply on the goods currently being repriced and profitable supply on the rest. Splitting the network this way is what lets the local balanced-flow analysis run while the seller's global requirements still hold. At termination, a lower-bounded flow chooses quantities between these upper and lower supplies, reconciling buyer clearing with seller profit maximization.

Finally, because marginal costs increase in discrete steps, the algorithm must adjust prices continuously without skipping over an equilibrium price. We exploit the structure of the cost functions: at the \emph{kinks} between consecutive cost segments, prices can be increased continuously while holding supply fixed. Within these fixed-supply intervals the DPSV-style balanced-flow machinery can be adapted, while the hybrid invariant handles transitions between intervals.

We next turn to the applications.

\paragraph{The Icelandic Arctic Auction.}
The Icelandic government explored the use of the Arctic auction as part of its strategy to exit capital controls following the financial crisis. In mid-2015, Iceland proposed holding an auction in which holders of previously blocked ``offshore'' krona accounts (foreign-owned ISK positions) could exchange their trapped krona balances for alternative financial assets such as cash or various government bonds in ISK or euros. The Arctic auction was designed in order to allow account holders to express their preferences over these alternative asset bundles in a single mechanism, with prices emerging endogenously, thereby efficiently clearing the blocked positions and facilitating capital market normalization. The Icelandic government expressed the desire to adjust the quantities and composition of assets supplied in response to participants' bids, subject to its own supply preferences and costs. Moreover, it wished to retain the option of exerting monopsonistic power, and so did not limit itself to solving the auction for competitive equilibrium. However, after political turmoil in 2016 triggered by the Panama Papers scandal and changes in government, the auction was dropped and most offshore accounts were eventually unblocked through alternative measures. 

\paragraph{The IMF Arctic Auction.}
The Arctic auction proposed by IMF staff for sovereign debt restructuring was motivated by long-standing failures in debt renegotiation, particularly the slow, costly, and conflict-prone negotiations that arise when a distressed country must bargain with many heterogeneous creditors. Traditional approaches (whether ad hoc negotiations, exchange offers, or legal mechanisms such as collective action clauses) often suffer from holdout problems, poor information about creditor preferences, and inefficient debt outcomes. Purely market-based solutions and formal statutory mechanisms have, so far, not succeeded in producing restructurings that are both timely and efficient, especially when creditors differ in risk tolerance, maturity preferences, and currency exposure.

To address these problems, the IMF-facilitated product-mix auction elicits creditors' preferences over a menu of restructured debt instruments (which differ, for example, in maturity, currency, or seniority) while also allowing the sovereign to specify the supply costs or trade-offs associated with issuing different instruments. By making these costs explicit, the auction ensures that creditor demand is matched against the sovereign's constraints and policy objectives, rather than being treated as exogenous. Letting prices and quantities be determined endogenously then allows the post-restructuring debt profile to reflect both creditor preferences and the sovereign's relative cost of supplying different claims, subject to an overall debt-relief constraint set by the IMF program. This approach reduces bargaining frictions, limits strategic holdouts, speeds up restructurings, and produces outcomes that are both more efficient and more politically feasible than negotiated exchanges—while preserving the IMF's role as a neutral coordinator rather than a direct counterparty.

\subsection{Related Work}
\label{sec:related}

The Product-Mix Auction by \citet{Klemperer2,Klemperer1,Klemperer-2008}, and its specialized implementation, the Arctic auction, have been widely studied in the context of central bank liquidity operations \cite{Fisher2011pricing}. Early algorithms for the Arctic auction \cite{Fichtl-computing} focused on very small instances. 

\citet{Baldwin2019understanding} studied the existence of equilibria in auctions and markets with indivisible goods, introducing the notion of \emph{demand types} to describe solutions when standard assumptions fail. Building on this, \citet{BKL2024} developed polynomial-time algorithms to compute competitive equilibrium prices and allocations for bidders with strong-substitutes valuations, leveraging submodular minimization for price determination and a novel constrained matching approach for supply allocation that directly uses the auction's explicit bidding language. Complementing these results, \citet{baldwin2024implementing} showed that the Product-Mix Auction's simple geometric bidding language uniquely captures all concave and strong-substitutes valuations, providing a new characterization of these preference classes and guaranteeing that the auction implements competitive equilibrium allocations.

In quasi-linear, budget-constrained multi-good markets, \cite{Lock-Finster2023competitive} shows that competitive equilibrium prices are simultaneously revenue- and constrained-welfare maximizing. While their work focuses on equilibrium structure and does not consider costs, our work develops a polynomial-time algorithmic framework for computing competitive equilibrium in the Arctic auction with stepwise-increasing separable costs. The flow-based techniques of \cite{DPSV} and its extensions \cite{Orlin-Fisher, Duan2015combinatorial, JV-Eisenberg, BeiGH16, VY-2025, Duan-GM} form a key algorithmic foundation for many of these results.

Our contribution builds on this literature by providing a polynomial-time algorithm for the Arctic auction with stepwise-increasing separable costs, integrating polyhedral theory techniques, flow-based methods, and quasi-linear preferences.

\subsection{Notation}

Vectors are written in bold. We write $[k]$ for the set $\{1, 2, \ldots, k\}$, and $\eb^i$ for the $i$th standard basis vector. We use $\Q$, $\Q_{\geq 0}$, and $\Q_{>0}$ to denote the rationals, the non-negative rationals, and the strictly positive rationals respectively, and analogously use $\R$, $\R_{\geq 0}$, and $\R_{>0}$ for the reals. We use the convention $1/\infty=0$.

\section{The Arctic Auction with Costs}
In the Arctic auction, bidders submit one or more \textit{Arctic bids} consisting of a budget and a value vector for the set $G$ of $n$ goods. Each such bid has the same demand as a buyer in the Fisher market with quasi-linear utilities, and the total demand of a bidder in the Arctic auction is then given by the (Minkowski) sum of the demands of her individual bids. Accordingly, throughout the paper we let $B$ denote the set of individual Arctic bids, and view these as buyers in the Fisher market. Let $m = |B|$. After computing an allocation in this market, the allocation to an original Arctic bidder is recovered by summing the allocations to her individual bids.

Formally, we assume that each buyer $i \in B$ has a value $u_{ij} \in \Q_{\geq 0}$ for good $j \in G$, and a budget $m_i \in \Q_{>0}$. Her utility $u_i(\xb_i, \pb)$ for bundle $\xb_i = (x_{ij})_{j \in G} \in \R^G_{\geq 0}$ of goods at prices $\pb \in \R^G_{>0}$ is
\[
    u_i(\xb_i, \pb) =
    \begin{cases}
        m_i + \sum_{j \in G} u_{ij} x_{ij} - \sum_{j \in G} p_{j} x_{ij}, & \text{if } \sum_{j \in G} p_j x_{ij} \leq m_i,\\
        -\infty, & \text{otherwise}.
    \end{cases}
\]
Thus the buyer always prefers a bundle that does not exceed her budget.

Throughout, we will restrict attention to strictly positive price vectors. In our cost setting, this will be ensured by the standing assumption that the first marginal cost of every good is strictly positive.

\begin{definition}
The \textit{bang per buck} of buyer $i$ for good $j$ at price $p_j$ is $\frac{u_{ij}}{p_j}$. The buyer's \textit{maximum bang per buck (mbpb)} at prices $\pb$ is
\[
    \alpha_i \coloneqq \max_{j \in G} \frac{u_{ij}}{p_j}.
\]
The goods that give the buyer a bang per buck of $\alpha_i$ are her \textit{maximum bang per buck goods} (or \textit{mbpb goods}) at $\pb$.
\end{definition}
Given fixed prices $\pb$, each buyer's utility $u_i(\xb_i, \pb)$ is linear in $\xb_i$ over the affordable region $\{ \xb_i \in \R^G_{\geq 0} \mid \sum_{j \in G} p_j x_{ij} \leq m_i \}$. Hence, any bundle $\xb_i$ which maximizes the utility of buyer $i$ at prices $\pb$ contains only mbpb goods, i.e., $x_{ij} > 0$ implies $\frac{u_{ij}}{p_j} = \alpha_i$. Moreover, the buyer may choose not to spend all her money precisely when $\alpha_i \leq 1$. For instance, if $u_{ij} < p_j$ for all goods $j \in G$, then it is immediate that $u_i(\xb_i, \pb)$ is maximized by the empty bundle $\bm{0}$. It follows that the buyer spends all her budget $m_i$ if $\alpha_i > 1$, a portion of her budget if $\alpha_i = 1$, and none of her budget if $\alpha_i < 1$.

For any buyer $i$ and prices $\pb$, define the set $M_i(\pb) \coloneqq \{j \in G \mid \frac{u_{ij}}{p_j} = \alpha_i \geq 1 \}$ containing the buyer's mbpb goods at $\pb$ if $\alpha_i \geq 1$ and the empty set $\emptyset$ otherwise.

\begin{lemma}[{\citep[Lemma 1]{Lock-Finster2023competitive}}]
\label{lemma:buyer-demand}
Fix a buyer $i \in B$ and prices $\pb$, and let $D^i(\pb)$ be the set of bundles the buyer demands at $\pb$.
\begin{enumerate}
    \item If $\alpha_i > 1$, then $D^i(\pb) = \conv \left \{ \frac{m_i}{p_j} \eb^j \mid j \in M_i(\pb) \right \}$.
    \item If $\alpha_i = 1$, then $D^i(\pb) = \conv \left( \{ \bm{0} \} \cup \left \{ \frac{m_i}{p_j} \eb^j \mid j \in M_i(\pb) \right \} \right )$.
    \item If $\alpha_i < 1$, then $D^i(\pb) = \{ \bm{0} \}$.
\end{enumerate}
\end{lemma}

An \textit{outcome} $(\pb, \xb)$ consists of a price $p_j$ for each good $j \in G$ and an allocation $x_{ij}$ of good $j$ to buyer $i$, for every good $j \in G$ and every buyer $i \in B$. An outcome is \textit{feasible} if each buyer receives a utility-maximizing bundle, that is, a bundle they demand. For convenience, we define $y_j = \sum_{i \in B} x_{ij}$ as the total quantity of good $j$ sold by the auctioneer.

The auctioneer has separable supply costs given by a stepwise increasing marginal cost function for each good.

\begin{definition}
\label{def:stepwise-marginal-costs}
The auctioneer's production costs are given by the separable function $c(\yb) = \sum_{j \in G} c_j(y_j)$. For each good $j \in G$, the function $c_j$ is specified by a sequence of $K_j$ \textit{breakpoints} $b_{j1} < \cdots < b_{j,K_j}$ and \textit{marginal costs} $0 \leq d_{j1} < \cdots < d_{j,K_j}$. The $k$th segment is the interval $(b_{j,k-1}, b_{jk}]$, on which $c_j$ has constant marginal cost $d_{jk}$; supplying beyond $b_{j,K_j}$ is infeasible. For notational convenience, we define $b_{j0} = 0$, $d_{j,K_j+1} = \infty$, and $c_j(0) = 0$. For $y_j>0$, define
\[
    c_j(y_j) =
    \begin{cases}
        \displaystyle\sum_{k=1}^{s-1} d_{jk}(b_{jk} - b_{j,k-1}) + d_{js}(y_j - b_{j,s-1}), & \text{if } y_j \leq b_{j,K_j},\\
        \infty, & \text{if } y_j > b_{j,K_j},
    \end{cases}
\]
where $s$ is the unique index in $[K_j]$ such that $y_j \in (b_{j,s-1}, b_{js}]$. See \cref{fig:stepwise-marginal-costs} for an illustration.
\end{definition}

The last breakpoint $b_{j,K_j}$ is an explicit supply cap for good $j$. In applications, this cap can be chosen so that it never binds at equilibrium. Throughout, we assume that $d_{j1} > 0$ for every good $j \in G$. The case $d_{j1} = 0$ can be treated separately and is omitted here.%
\footnote{Assuming $d_{j2} > 0$, a good $j$ with $d_{j1} = 0$ can be handled by initializing its price to any rational $p_j \in (0, d_{j2})$ and its supply to $t_j = b_{j1}$, then running the algorithm unchanged. Since demand exceeds supply at this small initial price, the algorithm raises $p_j$ until equilibrium is reached, exactly as in the $d_{j1} > 0$ case. No equilibrium price is skipped because the algorithm only needs to discover prices in $(0, d_{j2}]$ for the first segment.}

The auctioneer's profit from selling aggregate quantities $\zb \in \R^G_{\geq 0}$ at prices $\pb$ is his income from $\zb$ minus his cost,
\[
    \sum_{j \in G} p_j z_j - \sum_{j \in G} c_j(z_j).
\]

As the cost functions are separable, the auctioneer's profit-maximizing quantity for each good $j$ at fixed prices $\pb$ can be determined independently.

\begin{lemma}
\label{lemma:profit-maximising}
If $p_j < d_{j1}$, the auctioneer wants to sell $z_j=0$. If $p_j \in [d_{jk},d_{j,k+1})$ for some $k\in[K_j]$, then the profit-maximizing quantities are $z_j\in[b_{j,k-1},b_{jk}]$ when $p_j=d_{jk}$, and $z_j=b_{jk}$ when $p_j\in(d_{jk},d_{j,k+1})$.
\end{lemma}

Now that we have defined the buyers' demand and the auctioneer's profit, we can formulate competitive equilibrium.

\begin{definition}
An outcome $(\pb, \xb)$ (with aggregate quantities $\yb$) is a \textit{competitive equilibrium} if it is feasible and selling $\yb$ at $\pb$ maximizes the auctioneer's profit among all possible bundles $\zb \in \R^G_{\geq 0}$, so
\[
\yb \in \text{argmax}_{\zb \in \R^G_{\geq 0}} \sum_{j \in G} \left ( p_j z_j - c_j(z_j) \right ).
\]
\end{definition}
Note that the auctioneer's maximization problem in this definition is independent of the buyers' budgets, as $\zb$ can be any bundle of goods; quantities above the cap $b_{j,K_j}$ are simply infeasible because they incur infinite cost.

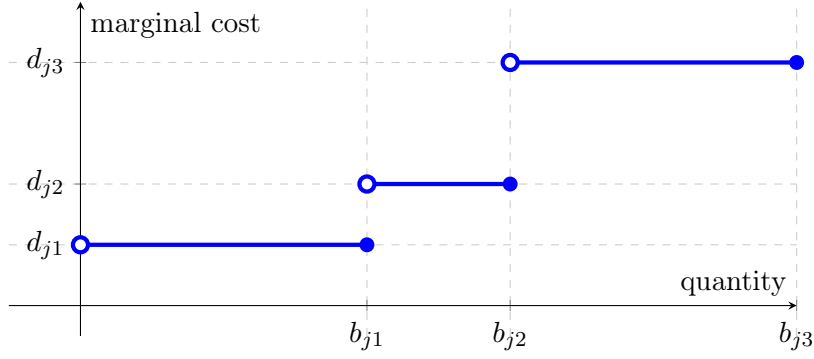
\begin{figure}[!t]
    \centering
    \begin{tikzpicture}
    \begin{axis}[
        axis lines=middle,
        xlabel={quantity},
        ylabel={marginal cost},
        xmin=-0.5,
        xmax=5,
        ymin=-0.5,
        ymax=5,
        xtick={0,2,3,5},
        ytick={1,2,4},
        xticklabels={$b_{j0}$,$b_{j1}$,$b_{j2}$, $b_{j3}$},
        yticklabels={$d_{j1}$,$d_{j2}$,$d_{j3}$},
        grid=major,
        grid style={dashed, gray!40},
        width=12cm,
        height=6cm,
    ]
    % Three horizontal line segments: open dot at start, closed dot at end
    \addplot[blue, ultra thick, mark indices={2}, mark=*, mark options={fill=blue}] 
        coordinates {(0,1) (2,1)} node[pos=0, circle, draw=blue, fill=white, inner sep=2pt]{};
    \addplot[blue, ultra thick, mark indices={2}, mark=*, mark options={fill=blue}] 
        coordinates {(2,2) (3,2)} node[pos=0, circle, draw=blue, fill=white, inner sep=2pt]{};
    \addplot[blue, ultra thick, mark indices={2}, mark=*, mark options={fill=blue}] 
        coordinates {(3,4) (5,4)} node[pos=0, circle, draw=blue, fill=white, inner sep=2pt]{};
    \end{axis}
    \end{tikzpicture}
    \caption{An illustration of a stepwise increasing marginal cost curve $c'_j$ for good $j$ with three segments, given by breakpoints $b_{j1}, b_{j2}, b_{j3}$ and segment heights $d_{j1}, d_{j2}, d_{j3}$.}
    \label{fig:stepwise-marginal-costs}
\end{figure}

\subsection{Rational-Value Property of Competitive Equilibrium}
Our first result establishes that rational inputs guarantee a rational-valued competitive equilibrium, a property that is not immediate with seller costs. Specifically, equilibrium prices and allocation quantities are rational numbers that can be represented by bit strings whose length is polynomial in the length required to represent the utilities, budgets, and cost function coefficients. We state this formally in the following theorem.

\begin{theorem}
\label{thm:rational-CE}
Suppose the utilities and budgets of the buyers, as well as the coefficients of the cost functions, are rational numbers. If a competitive equilibrium exists, there exists a competitive equilibrium that is rational-valued, with prices and allocation quantities having bit size polynomial in the bit size of the input parameters.
\end{theorem}

Our polynomial-time algorithm (\cref{thm:running-time}) in fact constructs such a rational-valued equilibrium, so \cref{thm:rational-CE} also follows from our algorithmic results. We nonetheless give a direct, self-contained proof via polyhedral theory, as the argument is short and may be of independent interest.

In the proof, we fix an arbitrary competitive equilibrium $(\pb^*,\xb^*)$ and record only its combinatorial type: each buyer's mbpb goods, whether her maximum bang per buck is strictly greater than $1$, and the cost segment containing the aggregate quantity sold of each good. In inverse-price variables $q_j=1/p_j$, these choices define a rational polyhedron whose constraints require allocated goods to remain mbpb, buyers to satisfy the spending conditions in \cref{lemma:buyer-demand}, and aggregate supplies to remain in the corresponding seller-optimal intervals. The original equilibrium gives a feasible point with $q_j>0$ for every good, and every feasible point with positive inverse prices gives a competitive equilibrium of the same type. Adding a variable $\epsilon$ and constraints $q_j\geq\epsilon$ lets us optimize for a positive lower bound on all inverse prices over a rational polyhedron. Standard rational linear-programming bit-complexity bounds then give a basic feasible solution with polynomially many bits. We now give the precise construction.

To construct the polyhedron $P$, first fix an arbitrary competitive equilibrium $(\pb^*, \xb^*)$. We may assume without loss of generality that $p_j^* \geq d_{j1}$ for every good $j \in G$: if $p_j^* < d_{j1}$ then the auctioneer's profit is maximized by selling zero of good $j$, so $y_j^* = 0$, and replacing $p_j^*$ by $d_{j1}$ yields another competitive equilibrium with the same allocation and all the same buyer demands (no buyer spends on a good with bang per buck below $1$, and raising $p_j^*$ cannot increase any buyer's bang per buck for good $j$). We do not assume that the entries of this equilibrium are rational, and this choice is used only for a nonconstructive existence proof of a rational equilibrium of the same combinatorial type. For each buyer $i \in B$, recall that $M_i(\pb^*)$ is the set of her mbpb goods at $\pb^*$ if her bang per buck at $\pb^*$ is at least $1$, and $M_i(\pb^*) = \emptyset$ otherwise. Also define $\overline{\alpha}^*_i = \max \{1, \max_{j \in G} \frac{u_{ij}}{p_j^*} \}$ for each buyer $i \in B$. Moreover, let $B_> = \{i \in B \mid \overline{\alpha}^*_i > 1 \}$ be the set of buyers whose maximum bang per buck is strictly greater than $1$. Finally, for each good $j \in G$, let $k_j$ be the smallest index in $[K_j]$ for which the quantity $\sum_{i \in B} x^*_{ij}$ of good $j$ aggregately supplied by the auctioneer satisfies $\sum_{i \in B} x^*_{ij} \le b_{j,k_j}$.

The variables of $P$ will be $q_j \in \R$ for each good $j \in G$, $\overline{\alpha}_i$ for each buyer $i \in B$, and $x_{ij} \in \R$ for each $i \in B$ and $j \in G$. Here $q_j$ should be interpreted as the inverse of the price of good $j$, and from feasible points with $q_j>0$ for every good we recover $p_j:= \frac{1}{q_j}$. We now construct $P$ by adding constraints for each buyer, as well as for the auctioneer.

For each buyer $i \in B$, we add constraints to ensure that she receives a demanded bundle at the given prices:
\begin{equation}
    \begin{aligned}
        q_j u_{ij}      &= \overline{\alpha}_i              && \forall j \in M_i(\pb^*),               \\
        q_j u_{ij}      &\leq \overline{\alpha}_i           && \forall j \in G \setminus M_i(\pb^*),   \\
        x_{ij}          &\geq 0                             && \forall j \in G,                         \\
        x_{ij}          &= 0                                && \forall j \in G \setminus M_i(\pb^*),   \\
        \overline{\alpha}_i        &\geq 1                             &&                                  \\
        \overline{\alpha}_i        &=1                                 && \text{if } i \notin B_>       \\
        \overline{\alpha}_i m_i    &\geq \sum_{j \in M_i(\pb^*)} u_{ij}x_{ij} &&                                  \\
        \overline{\alpha}_i m_i    &= \sum_{j \in M_i(\pb^*)} u_{ij}x_{ij}    && \text{if } i \in B_>
    \end{aligned}
\end{equation}

For the auctioneer, we add constraints that ensure he maximizes profit of good $j$ at price $\frac{1}{q_j}$. Note that the quantity $\sum_{i \in B} x_{ij}$ of each good $j$ aggregately supplied by the auctioneer is either equal to the breakpoint $b_{j,k_j}$, or it lies in $[b_{j,k_{j}-1}, b_{j,k_j})$. For each good $j \in G$, we add one of the following constraints, depending on the case. If $\sum_{i \in B} x^*_{ij} = b_{j,k_j}$, then add constraints
\begin{align*}
\frac{1}{d_{j,k_{j}+1}} \leq q_j \leq \frac{1}{d_{j,k_j}}, \\
\sum_{i \in B} x_{ij} = b_{j,k_j}.
\end{align*}
If $\sum_{i \in B} x^*_{ij} \in [b_{j,k_{j}-1}, b_{j,k_j})$, then add constraints
\begin{align*}
q_j = \frac{1}{d_{j,k_j}},\\
b_{j,k_{j}-1} \leq \sum_{i \in B} x_{ij} \leq b_{j,k_j}.
\end{align*}

The last inequalities follow because when price is $d_{j,k_j}$ then the auctioneer is willing to sell at least $b_{j,k_{j}-1}$ amount of good $j$ and at most $b_{j,k_j}$ amount of good $j$. We note that although the polyhedron $P$ was defined based on the competitive equilibrium $(\pb^*, \xb^*)$, its coefficients are rational, and we refer to it as a rational polyhedron.

As the following lemma shows, $(\qb^*, \xb^*, \overline{\alpha}^*)$ is a feasible point of $P$, and every feasible point of $P$ with positive inverse-price coordinates yields a competitive equilibrium. Moreover, $(\qb^*, \xb^*, \overline{\alpha}^*)$ has $q_j^*>0$ for every good. Since $P$ is rational and the open set $P\cap\{q_j>0\ \forall j\}$ is non-empty, introduce a variable $\epsilon$, add $0\leq \epsilon \leq 1$ and $q_j\geq \epsilon$ for all goods $j$, and maximize $\epsilon$ over this rational polyhedron. The optimum is positive, and the standard rational linear-programming bit-complexity bound gives an optimal basic feasible solution with polynomially many bits \citep{Sch-book}. Dropping $\epsilon$ yields a rational feasible point of $P$ with $q_j>0$ for every good. This completes the proof of \cref{thm:rational-CE}.

\begin{lemma}
\label{lemma:polyhedron-CE}
Let $q^*_j = \frac{1}{p^*_j}$ for each good $j \in G$, and let $\overline{\alpha}^*_i = \max \{1, \max_{j \in G} q^*_j u_{ij}\}$ for each buyer $i \in B$. Then $(\qb^*, \xb^*, \overline{\alpha}^*)$ is a feasible point in the polyhedron $P$.

Conversely, if $(\qb, \xb, \overline{\alpha})$ is a feasible point of $P$ with $q_j>0$ for every good $j$ and we define $p_j = \frac{1}{q_j}$, then $(\pb, \xb)$ is a competitive equilibrium.
\end{lemma}
\begin{proof}
The first claim follows directly from the fact that $(\pb^*, \xb^*)$ is a competitive equilibrium, together with the definitions of $M_i(\pb^*)$, $B_>$, and $k_j$, \cref{lemma:buyer-demand}, and \cref{lemma:profit-maximising}.

Now fix a feasible point $(\qb, \xb, \overline{\alpha})$ of $P$ with $q_j > 0$ for all $j \in G$, and define $p_j = \frac{1}{q_j}$ for all goods $j \in G$. We first show that every buyer $i \in B$ demands bundle $\xb_i$ at prices $\pb$. By the constraints of $P$, we have
\begin{align*}
    \frac{u_{ij}}{p_j} = q_j u_{ij} = \overline{\alpha}_i &\qquad \forall j \in M_i(\pb^*), \\
    \frac{u_{ij}}{p_j} = q_j u_{ij} \le \overline{\alpha}_i &\qquad \forall j \in G \setminus M_i(\pb^*).
\end{align*}
If $M_i(\pb^*) \neq \emptyset$, the first line implies that some good attains bang per buck $\overline{\alpha}_i$. If $M_i(\pb^*) = \emptyset$, then by definition of $M_i(\pb^*)$ we have $i \notin B_>$. Noting that $\overline{\alpha}_i = 1$ and $\frac{u_{ij}}{p_j} \le \overline{\alpha}_i$ for all $j \in G$, in both cases we thus get
\[
    \overline{\alpha}_i = \max \left\{ 1, \max_{j \in G} \frac{u_{ij}}{p_j} \right\}.
\]
Moreover, the constraints $x_{ij} = 0$ for all $j \notin M_i(\pb^*)$ imply that every good allocated to buyer $i$ satisfies $\frac{u_{ij}}{p_j} = \overline{\alpha}_i$. Hence $\xb_i$ contains only goods maximizing buyer $i$'s bang per buck at prices~$\pb$.

For the spending conditions, note that $x_{ij} = 0$ outside $M_i(\pb^*)$, and for every $j \in M_i(\pb^*)$ we have $u_{ij} = \overline{\alpha}_i p_j$. Therefore,
\[
    \sum_{j \in G} p_j x_{ij} = \frac{1}{\overline{\alpha}_i} \sum_{j \in G} u_{ij} x_{ij}.
\]
Using the last two buyer constraints of $P$, it follows that
\[
    \sum_{j \in G} p_j x_{ij} \le m_i,
\]
with equality whenever $i \in B_>$. Since $\overline{\alpha}_i = 1$ for every buyer $i \notin B_>$, these are precisely the spending conditions required by \cref{lemma:buyer-demand}. Thus buyer $i$ demands bundle $\xb_i$ at prices $\pb$.

Finally, let $y_j = \sum_{i \in B} x_{ij}$ for each good $j \in G$. If $y^*_j = b_{j,k_j}$, then feasibility of $P$ gives
\[
    p_j \in [d_{j,k_j}, d_{j,k_j+1}]
    \qquad \text{and} \qquad
    y_j = b_{j,k_j}.
\]
If $p_j < d_{j,k_j+1}$, then this is profit-maximizing for good $j$ by \cref{lemma:profit-maximising}; this includes the last-segment case $k_j=K_j$, where $d_{j,K_j+1}=\infty$. If $p_j = d_{j,k_j+1}$ and $k_j < K_j$, then the seller is indifferent over every quantity in $[b_{j,k_j}, b_{j,k_j+1}]$, so $y_j = b_{j,k_j}$ is again profit-maximizing. If instead $y^*_j \in [b_{j,k_j-1}, b_{j,k_j})$, then feasibility of $P$ gives
\[
    p_j = d_{j,k_j}
    \qquad \text{and} \qquad
    y_j \in [b_{j,k_j-1}, b_{j,k_j}],
\]
which is again profit-maximizing for good $j$ by the same characterization. Hence the auctioneer maximizes profit with $\zb$ at $\pb$, and therefore $(\pb, \xb)$ is a competitive equilibrium.
\end{proof}

\section{Solving Arctic Auctions with Stepwise Increasing Marginal Costs}

The existence of a rational-valued competitive equilibrium motivates the design of a polynomial-time algorithm to find one. We propose a primal-dual algorithm for solving the Arctic auction with stepwise increasing marginal costs, building upon the work of  \cite{DPSV,arctic_markets_production} for linear Fisher markets and Arctic auctions without costs. However, our algorithm requires several key extensions. In fixed-supply markets, the same supply vector both limits the flow during the algorithm and describes seller-side feasibility at equilibrium. With stepwise increasing marginal costs these roles separate: actual supplies are upper bounds during the algorithm, while seller optimality at final prices imposes profitable lower bounds. We handle this by maintaining a hybrid min-cut invariant and recovering the final allocation through a lower-bounded flow.

Throughout, we assume that $d_{j1} > 0$ for every good $j \in G$. Since the algorithm initializes prices to $p_j = d_{j1}$ and only increases them, all prices maintained by the algorithm are strictly positive. In particular, this ensures that the bang-per-buck ratios $u_{ij}/p_j$ are well-defined throughout.

\subsection{Preliminaries}

We use standard network-flow terminology throughout this section. The definitions and basic facts needed for the algorithm, including residual graphs, min-cuts, lower-bounded flows, and Hoffman's circulation criterion, are collected in Appendix~\ref{app:network-flow-background};  see also \citet{AhujaNetworkFlows} for further details.

\begin{definition}
For any prices $\pb \in \Q_{>0}^G$, budgets $\mb \in \Q_{>0}^B$, and supply $\tb \in \Q_{\geq 0}^G$, let the \textit{expenditure network} $N(\pb, \mb, \tb)$ be the flow network with vertices consisting of a source $s$, a sink $t$, the buyers $B$, and the goods $G$. These vertices are connected by
\begin{itemize}
\item an arc from the source $s$ to each good $j$ with capacity $\capac(s, j) = p_jt_j$;
\item an arc from each buyer $i$ to the sink $t$ with capacity $\capac(i, t) = m_i$;
\item an arc from good $j$ to buyer $i$ with capacity $\capac(j,i) = \infty$ if $j$ is a mbpb good for buyer $i$ at prices $\pb$.
\end{itemize}
We illustrate this network in \cref{fig:expenditure-flow-network}.
\end{definition}

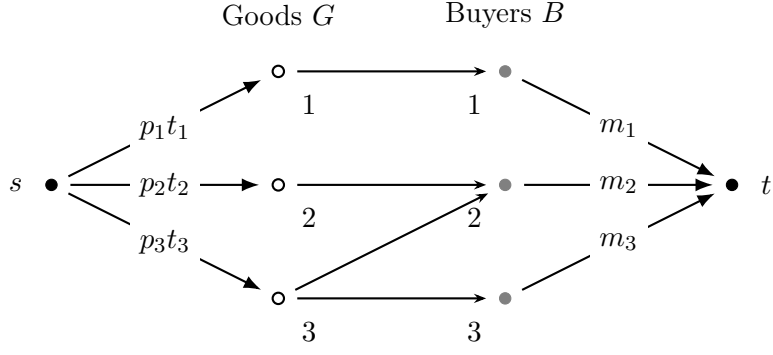
\begin{figure}
\centering

\begin{tikzpicture}[
    xscale=3,
    yscale=1.5,
    buyer/.style={circle, draw=gray, fill=gray, minimum size=0.15cm, inner sep=0pt, outer sep=5pt},
    good/.style={circle, draw, thick, minimum size=0.15cm, inner sep=0pt, outer sep=5pt},
    vx/.style={circle, draw=black, fill=black, minimum size=0.15cm, inner sep=0pt, outer sep=5pt},
    arc/.style={thick, draw=black, -{Stealth[scale=0.7]}},
    foo/.style={coordinate}
]

% Source and sink
\node[vx, label=left:{$s$}] (s) at (0,1) {};
\node[vx, label=right:{$t$}] (t) at (3,1) {};

% Goods
\node[good, label=below right:$1$] (g1) at (1,2) {};
\node[good, label=below right:$2$] (g2) at (1,1) {};
\node[good, label=below right:$3$] (g3) at (1,0) {};

% Buyers
\node[buyer, label=below left:$1$] (b1) at (2,2) {};
\node[buyer, label=below left:$2$] (b2) at (2,1) {};
\node[buyer, label=below left:$3$] (b3) at (2,0) {};

% Connections buyers to sink
\draw[trade] (b1) to node[midway,fill=white] {$m_1$} (t);
\draw[trade] (b2) to node[midway,fill=white] {$m_2$} (t);
\draw[trade] (b3) to node[midway,fill=white] {$m_3$} (t);

% Connections goods to buyers
\path[arc] (g1) edge (b1);
\path[arc] (g2) edge (b2);
\path[arc] (g3) edge (b2);
\path[arc] (g3) edge (b3);

% Connections source to goods
\draw[trade] (s) to node[midway,fill=white] {$p_1 t_1$} (g1);
\draw[trade] (s) to node[midway,fill=white] {$p_2 t_2$} (g2);
\draw[trade] (s) to node[midway,fill=white] {$p_3 t_3$} (g3);

% Goods and Buyers labels
\node[] (goods) at (1, 2.5) {Goods $G$};
\node[] (buyers) at (2, 2.5) {Buyers $B$};

\end{tikzpicture}
\caption{The expenditure network $N(\pb,\mb,\tb)$ with the arcs labeled with their capacities. The arcs from goods to buyers have capacity $\infty$.}
\label{fig:expenditure-flow-network}
\end{figure}

A key ingredient in the algorithm is the notion of a balanced flow. Intuitively, this is a flow that distributes the surplus (the unspent money) of the buyers as equitably as possible.

\begin{definition}
The \textit{surplus} $\bm{\gamma}(f)$ of flow $f$ in network $N(\pb, \mb, \tb)$ is given by $\gamma(f)_i = m_i - f(i, t)$ for each buyer $i \in B$.
A \textit{balanced flow} of network $N(\pb, \mb, \tb)$ is a max flow of $N(\pb, \mb, \tb)$ that minimizes $\| \bm{\gamma}(f) \|_2^2$.
\end{definition}

For a feasible flow $f$ in a network $N$, we write $R_N(f)$ for the \emph{restricted residual graph}: the residual graph restricted to the buyer and good vertices $B \cup G$, with $s$, $t$, and all their incident arcs removed. We will also use the following standard residual-graph property of balanced flows from \citet{DPSV}.

\begin{lemma}
\label{lemma:balanced-flow-residual}
Let $N = N(\pb,\mb,\tb)$, let $f$ be a balanced flow of $N$, and let $i,i' \in B$. If $\gamma(f)_i < \gamma(f)_{i'}$, then there is no directed path from $i$ to $i'$ in the restricted residual graph $R_N(f)$.
\end{lemma}

It is well-known that balanced flows can be computed in polynomial time. \citet{DPSV} designed a divide-and-conquer method that uses a linear number of max flow computations to find a balanced flow. Later, \citet{darwish2016improved} introduced a parametric flow one-shot method for computing balanced flows using $O(\log |V(N)|)$ max flow computations. Since the expenditure network has $|V(N)| = m+n+2$ vertices, this gives $O(\log(m+n))$ max flow computations in our setting.

Next, we introduce the notion of \textit{tight sets}. For any subset $S \subseteq B$ of buyers, let $\Gamma(S)$ be the set of goods adjacent to $S$ in $N(\pb, \mb, \tb)$. Similarly, for any subset $S \subseteq G$ of goods, let $\Gamma(S)$ be the set of buyers adjacent to $S$ in $N(\pb, \mb, \tb)$.

\begin{definition}[Tight sets]
A subset $S \subseteq G$ of goods is \textit{tight} if
\[
    \sum_{j \in S} p_j t_j \geq \sum_{i \in \Gamma(S)} m_i.
\]
\end{definition}

\subsection{An Overview of the Algorithm}
The algorithm combines the fixed-supply balanced-flow machinery from the Arctic auction without costs with explicit updates to supplies and buyers' remaining budgets. In the primal steps, the algorithm scales up a judiciously chosen subset of goods by a common factor, causing prices to increase monotonically. In the dual steps, the algorithm computes a {balanced flow} in the {expenditure network}, in which each buyer is connected only to her current mbpb goods. When a buyer's maximum bang per buck drops to $1$, her budget is partially or fully returned; in the latter case, she is removed from the market entirely.

Formally, let $\widehat B$ denote the original buyer set of the instance. On the buyer side, the algorithm maintains the active buyer set $B \subseteq \widehat B$ and current budgets $(m_i)_{i \in B}$, initialized from the input; the expenditure network $N(\pb, \mb, \tb)$ is always built only on the active buyers, and removed buyers receive zero allocation in the final output.

On the seller side, the algorithm maintains a segment index $\seg_j$ for each good $j$, recording which segment of the cost function the current price lies on, and a supply vector $\tb$: each good's supply is initialized to the maximum the seller is willing to sell at the starting price and is increased only when the price reaches the next marginal-cost breakpoint, ensuring the auctioneer never sells more than is profitable at current prices. The proof simultaneously tracks a profitable supply vector that records the lower quantities needed for seller optimality.

\begin{definition}[Profitable supply]
\label{def:profitable-supply}
For any current prices $\pb$ and cost segment indices $(\seg_j)_{j\in G}$ maintained by the algorithm, define the \emph{profitable supply} vector $\bar{\tb}$ by
\[
    \bar t_j =
    \begin{cases}
        b_{j,\seg_j-1} & \text{if } p_j = d_{j,\seg_j},\\
        b_{j,\seg_j} & \text{if } p_j \in (d_{j,\seg_j}, d_{j,\seg_j+1}),
    \end{cases}
    \qquad \text{for every } j \in G,
\]
Thus $\bar t_j$ is the smallest quantity of good $j$ that the seller must sell to maximize profit at current prices.
\end{definition}

We refer to $N(\pb,\mb,\bar{\tb})$ as the \emph{profitable network} at prices $\pb$.

The full algorithm is given in \cref{alg:arctic-with-costs}. We now highlight its key features.

\paragraph{Initialization.}
Note that the auctioneer's stepwise increasing marginal cost curves specify a minimum price $d_{j1}$ at which he is willing to sell any quantity of each good $j \in G$. Therefore, we initialize prices to $p_j = d_{j1}$. At these initial prices, the auctioneer maximizes profit by selling up to a quantity of $b_{j1}$ of good $j$ (by \cref{lemma:profit-maximising}), so supply is set to $t_j = b_{j1}$. Since buyers with a maximum bang per buck of less than $1$ demand nothing at the initial prices, and prices increase monotonically throughout the algorithm, these buyers are removed. If this empties $B$, the algorithm immediately returns zero allocation at prices $\pb$: every removed buyer demands the empty bundle since $\alpha_i < 1$, and the auctioneer is satisfied by zero supply since prices are at the first breakpoint ($\bar t_j = b_{j0} = 0$), so this is a competitive equilibrium. Otherwise the algorithm maintains the invariant that only buyers with maximum bang per buck at least $1$, at the current prices $\pb$, remain in the market.

\paragraph{The inner and outer loops.}
The main part of the algorithm consists of an outer and an inner loop. We call each execution of the inner loop body an \emph{inner iteration}, and a \emph{phase} consists of one iteration of the outer loop together with all the inner iterations it spawns. Each iteration of the outer loop begins by rebuilding the full expenditure network $N = N(\pb,\mb,\tb)$ at the current prices, budgets, and supply, and computing a balanced flow of this network. If this balanced flow has zero surplus for every buyer, then the outer loop terminates and enters the final processing step. Otherwise, we let $I$ be the set of buyers with the largest surplus, and let $J$ be the goods adjacent to $I$ in the expenditure network. The algorithm then removes all the arcs in $N$ from goods in $J$ to buyers not in $I$, and identifies all newly \emph{isolated} buyers $Z \subseteq B\setminus I$ in $N$. Here, a buyer is considered isolated if they have no incoming arcs from any good in $G$. We refer to the restricted subnetwork induced by $\{s\}\cup J\cup I\cup\{t\}$ after these deletions as the current \emph{phase component}.

The algorithm then enters the inner loop, which scales up the prices of goods $J$ by a common factor $\theta \geq 1$. Conceptually, we determine $\theta$ by continuously increasing it from $1$ upwards, until one of Events $1$ to $6$ occurs. Formally, we compute, for each event $e$, the smallest $\theta_e \geq 1$ for which the event holds, and set $\theta = \min_{e \in [6]} \{\theta_e \}$. The value $\theta_1$ may be infinite if every good in $J$ already lies on its last segment, and some of the other event values may also be infinite when the relevant set of buyers or goods is empty. Nevertheless, $\theta$ is well-defined: unless an earlier event occurs, the buyers in $I$ eventually reach maximum bang per buck $1$, triggering Event~6. In each inner iteration, the algorithm executes only the smallest-index event attaining this minimum scaling factor.
It is straightforward to compute the values $\theta_1, \ldots, \theta_4$ and $\theta_6$ in linear time. For Event~5, we need to compute the smallest scaling factor $\theta_5 \ge 1$ for which some nonempty subset $S \subseteq J$ becomes tight, that is, $\sum_{i \in \Gamma(S)} m_i \leq \sum_{j \in S} \theta_5 p_j t_j$. 

The value $\theta_5$ may equal $1$ immediately after an Event~2 restoration, in which case no positive price scaling is needed before Event~5 can end the phase. Otherwise the same subroutines as in \citet{DPSV} compute $\theta_5$ using at most $n$ max flow computations. If one also wants an explicit tight set $S$ (e.g., for analysis), it can be recovered with one additional max flow computation, but the algorithm itself does not require this. We now describe the events in more detail.

\begin{description}
    \item[Event 1] This event fires when the prices $p_j$ of one or more goods newly reach their next supply breakpoints. For each such good $j$ with current segment $\seg_j = k < K_j$ and price $p_j = d_{j,k+1}$, the auctioneer is willing to sell up to quantity $b_{j,k+1}$, so we increase the supply to $t_j = b_{j,k+1}$ and update the current segment to $\seg_j = k+1$. This may create temporary oversupply in the expenditure network, but that is harmless: source capacities are only upper bounds during the algorithm, and the final lower-bounded flow enforces the seller's actual profit-maximizing quantities.
    \item[Event 2] As the prices of goods $J$ increase, the goods not in $J$ become more desirable relative to those in $J$. This event fires when some buyer $i \in I$ starts demanding a good $j \in G \setminus J$. We add the arc $(j,i)$ to the network $N$ and recompute the balanced flow. If this balanced flow has zero surplus on every buyer, then the outer loop terminates and we enter the algorithm's final processing steps. Otherwise, we add to $I$ every buyer $i' \in B\setminus I$ that has a directed path to $I$ in the restricted residual graph $R_N(f)$. We update $J$ to be all the goods adjacent to $I$ in $N$, remove all arcs from goods in the new set $J$ to buyers in $B\setminus I$, and recompute the set $Z$ of isolated buyers. By \cref{lemma:event2-restoration}, the set $J$ strictly increases and every arc deleted in this restoration step carries zero flow in the recomputed balanced flow.
    \item[Event 3] Similarly to Event $2$, this event fires when some buyer $i \in Z$ starts demanding a good $j \in G \setminus J$. We add the arc $(j,i)$ to $N$ and remove $i$ from $Z$ (but not from $B$).
    \item[Event 4] As prices increase, the maximum bang per buck of buyers decreases monotonically. This event fires when a buyer $i \in Z$ reaches maximum bang per buck $\alpha_i = 1$. She becomes indifferent between receiving an empty bundle and receiving any bundle containing only mbpb goods, so we remove buyer $i$ from the market. Since prices will continue to rise, such a buyer will always be willing to accept an empty bundle in the future.
    \item[Event 5] This event fires when a nonempty subset $S \subseteq J$ of goods becomes tight. The current phase ends and a new phase begins. The algorithm does not need to identify a particular tight set explicitly, since the next outer-loop iteration rebuilds the full expenditure network from scratch.
    \item[Event 6] Similarly to Event~$4$, this event fires when a buyer $i \in I$ reaches $\alpha_i = 1$. We form the restricted phase network $N_{I,J}^{(i,0)}$ induced by $\{s\}\cup J\cup I\cup\{t\}$, with the current source capacities $p_j t_j$ on $J$, the current mbpb edges between $J$ and $I$, and buyer $i$'s sink capacity temporarily set to $0$. If the source-side cut $(\{s\}, J \cup I \cup \{t\})$ is a min-cut of $N_{I,J}^{(i,0)}$, we remove the buyer. Otherwise, we compute the unique inclusion-wise maximal source-side min-cut of $N_{I,J}^{(i,0)}$ with good side $S \subseteq J$ and buyer side $T \subseteq I$; uniqueness holds because the union of the source sides of source-side min-cuts is again the source side of a min-cut. We set
    \[
        m_i' = \sum_{j \in S} p_j t_j - \sum_{i' \in T \setminus \{i\}} m_{i'}
    \]
    and update $m_i \leftarrow m_i'$. Since Event~5 has priority over Event~6, \cref{lemma:no-tight-subset} shows that $\sum_{j \in S} p_j t_j < \sum_{i' \in \Gamma(S)} m_{i'}$ for every $S \subseteq J$ at the moment Event~6 fires. This implies $m_i' < m_i$, so the budget strictly decreases. Since the source-side cut has capacity $C = \sum_{j \in J} p_j t_j$ and is not a min-cut of $N_{I,J}^{(i,0)}$, the chosen min-cut has capacity $C - m_i' < C$, hence $m_i' > 0$, so the updated budget remains positive.
\end{description}

\paragraph{The Final Phase.}
The outer loop terminates once some balanced-flow computation yields zero surplus for every remaining buyer.
At that moment, the current prices already support full spending by all buyers on their mbpb goods. Before the final processing step, we rebuild the full expenditure network $N(\pb,\mb,\tb)$ at these final prices, budgets, and supply. As the supply of a good only increases once prices are sufficiently high, the auctioneer never sells more than they wish at $\pb$. However, at the final prices the seller also requires a minimum quantity of each good to be sold in order to maximize profit; this is exactly the final profitable supply $\bar{\tb}$ defined in \cref{def:profitable-supply}. The final processing step computes a feasible flow in the lower-bounded network $\overline N$ obtained from $N$ by giving each source arc $(s,j)$ lower bound $p_j \bar t_j$ and upper bound $p_j t_j$, each mbpb arc $(j,i)$ upper bound $\infty$, and each buyer arc $(i,t)$ lower and upper bounds both equal to $m_i$. \Cref{lemma:profitable-mincut,lemma:final-lb-feasible} will show that this lower-bounded network is feasible, so the returned allocation simultaneously exhausts every remaining buyer's budget and satisfies the seller's profit-maximizing lower bounds.

\begin{algorithm}[htb!]
\small
\begin{algorithmic}[1]
\Require Arctic auction instance $(\widehat B, G, \ub, \mb, \bb, \db)$.
\Ensure Competitive equilibrium $(\pb, \xb)$.
\State Set $B = \widehat B$.
\State For each good $j \in G$, set price $p_j = d_{j1}$, initial supply $t_j = b_{j1}$ and $\seg_j = 1$.
\State Remove all buyers $i$ with $\alpha_i = \max_{j \in G} \frac{u_{ij}}{p_j} < 1$ from $B$.
\If{$B = \emptyset$}
    \State Set $x_{ij} = 0$ for all $i \in \widehat B$, $j \in G$; \Return $(\pb, \xb)$.
\EndIf
\OuterLoop
    \State Build $N = N(\pb, \mb, \tb)$ and compute a balanced flow $f$.
    \If{$\max_{i \in B} \gamma(f)_i = 0$}
        \State \textbf{break} out of outer loop.
    \EndIf
    \State Let $I = \argmax_{i \in B} \gamma(f)_i$.
    \State Let $J = \Gamma(I)$, delete all edges from goods $J$ to buyers $B\setminus I$ in $N$, and let $Z = \{i \in B\setminus I \mid \Gamma(\{i\}) = \emptyset \}$.
   \InnerLoop
        \State Compute the smallest scaling factor $\theta \geq 1$ at which one of Events $1$--$6$ occurs.
        \State Scale up prices of goods in $J$ by setting $p_j = \theta p_j$.
        \If{\textbf{Event 1} ($p_j=d_{j,\seg_j+1}$ for some $j\in J$ with $\seg_j<K_j$)}
            \State Let $H = \{j \in J \mid \seg_j < K_j \text{ and } p_j = d_{j,\seg_j+1}\}$.
            \State For each $j \in H$, set $t_j = b_{j,\seg_j+1}$, $\capac(s,j) = p_j t_j$, and $\seg_j \leftarrow \seg_j+1$; \textbf{break} out of inner loop.
        \ElsIf{\textbf{Event 2} (a buyer $i \in I$ starts demanding a good $j \in G\setminus J$ at $\pb$) }
            \State Add arc $(j,i)$ with capacity $\capac(j,i) = \infty$ to $N$ and recompute the balanced flow $f$.
            \If{$\max_{i \in B} \gamma(f)_i = 0$}
                \State \textbf{break} out of both loops.
            \EndIf
            \State Let $I'$ be all buyers $i' \in B\setminus I$ with a directed path from $i'$ to $I$ in $R_N(f)$; set $I \leftarrow I \cup I'$, $J \leftarrow \Gamma(I)$, delete all edges from $J$ to $B\setminus I$, and recompute $Z$.
        \ElsIf{\textbf{Event 3} (a buyer $i \in Z$ starts demanding a good $j \in G\setminus J$ at $\pb$)}
            \State Add arc $(j,i)$ with capacity $\capac(j,i) = \infty$ to $N$ and remove $i$ from $Z$.
        \ElsIf{\textbf{Event 4} (a buyer $i \in Z$ becomes indifferent: $\alpha_i = 1$)}
            \State Set $B \leftarrow B\setminus\{i\}$ and delete $i$ and its incident arcs.
        \ElsIf{\textbf{Event 5} (a nonempty subset $S \subseteq J$ of goods goes tight: $\sum_{j \in S} p_j t_j \geq \sum_{i \in \Gamma(S)} m_i$)}
            \State \textbf{break} out of inner loop.
        \ElsIf{\textbf{Event 6} (a buyer $i \in I$ becomes indifferent: $\alpha_i = 1$)}
            \State Let $N_{I,J}^{(i,0)}$ be the restricted network on $\{s\}\cup J\cup I\cup\{t\}$, with buyer arc $(i,t)$ capacity set to $0$.
            \If{$(\{s\}, J \cup I \cup \{t\})$ is an $s$-$t$ min-cut in $N_{I,J}^{(i,0)}$}
                \State Set $B \leftarrow B\setminus\{i\}$ and delete $i$ and its incident arcs.
            \Else
                \State Find a maximal source-side min-cut with good side $S$ and buyer side $T$ in $N_{I,J}^{(i,0)}$, and set $m_i \leftarrow \sum_{j \in S} p_j t_j - \sum_{i' \in T \setminus \{i\}} m_{i'}$.
            \EndIf
            \State \textbf{break} out of inner loop.
        \EndIf
 \EndInnerLoop
\EndOuterLoop
\State Rebuild $N = N(\pb, \mb, \tb)$. For each good $j \in G$, compute $\bar \tb$ as $\bar t_j = b_{j,\seg_j-1}$ if $p_j = d_{j,\seg_j}$, and $\bar t_j = b_{j,\seg_j}$ if $p_j \in (d_{j,\seg_j}, d_{j,\seg_j+1})$.
\State Let $\overline N$ be the lower-bounded network with the same vertices and arcs as $N$, lower bound $p_j \bar t_j$ and upper bound $p_j t_j$ on each arc $(s,j)$, lower and upper bounds both $m_i$ on each arc $(i,t)$, and upper bound $\infty$ on each arc $(j,i)$. Find a feasible flow $f$ of $\overline N$.
\State For every $i \in \widehat B$ and $j \in G$, set $x_{ij} = f(j,i)/p_j$ if $(j,i)$ is an arc of $N$, and $x_{ij}=0$ otherwise.
\State \Return $(\pb, \xb)$.
\end{algorithmic}
\vspace{-0.5em}
\caption{Solving Arctic auctions with stepwise increasing marginal costs.}
\label{alg:arctic-with-costs}
\end{algorithm}

\subsection{Properties of the Algorithm}

We now establish six structural properties of the algorithm, which are used in the correctness proof (\cref{prop:correctness}) and the running-time analysis (\cref{thm:running-time}). The first two lemmas (\cref{lemma:event2-restoration,lemma:no-tight-subset}) control the combinatorial structure within a phase. The next three (\cref{lemma:source-cut-characterisation,lemma:hybrid-mincut,lemma:profitable-mincut}) together maintain the key invariant that the source-side cut is a min-cut of the profitable network throughout the algorithm. The final lemma (\cref{lemma:final-lb-feasible}) uses this invariant to show that the lower-bounded network in the final processing step is always feasible, which is the main technical ingredient in the correctness proof.

The first lemma studies the restoration step of Event~2, when a new mbpb edge is added and the phase component must be rebuilt.

\begin{lemma}
\label{lemma:event2-restoration}
A nonterminal Event~$2$ strictly increases the phase set $J$. Moreover, the deleted arcs carry zero flow in the newly computed $f$.
\end{lemma}
\begin{proof}
Let $I^-$ and $J^-$ be the phase sets just before Event~$2$, and let $f$ be the balanced flow computed after adding the new mbpb arc. The restoration step sets
\[
    I'=\{i'\in B\setminus I^- \mid \text{there is a directed path from } i' \text{ to } I^- \text{ in } R_N(f)\},
\]
then updates $I^+=I^-\cup I'$ and $J^+=\Gamma(I^+)$, and deletes the arcs from goods in $J^+$ to buyers in $B\setminus I^+$.

The newly added mbpb arc is the first edge from some good $j\in G\setminus J^-$ to a buyer in $I^-$. Since $I^-\subseteq I^+$, this good belongs to $\Gamma(I^+)=J^+$. Hence $J^+\supsetneq J^-$.

It remains to check the deletion step. Consider an arc $(g,z)$ deleted during the restoration, with $g\in J^+$ and $z\in B\setminus I^+$. If $f(g,z)>0$, then the reverse residual edge $z\to g$ is present in $R_N(f)$. Since $g\in J^+=\Gamma(I^+)$, there is a buyer $i\in I^+$ adjacent to $g$, and the forward residual edge $g\to i$ is present because goods-to-buyers arcs have infinite capacity. If $i\in I^-$ we have a path from $z$ to $I^-$; if $i\in I'$, we concatenate $z\to g\to i$ with the path from $i$ to $I^-$. In either case $z\in I'$, contradicting $z\in B\setminus I^+$. Hence every deleted arc has zero flow.
\end{proof}

The next lemma records where the phase component has strict Hall-style slack. Fresh phases start with strict slack. Later in a phase, Event~2 may create a tight set at the current prices; when Event~6 is executed, however, Event~5's priority guarantees that no such tight set is present.

\begin{lemma}
\label{lemma:no-tight-subset}
At the start of every new phase, every nonempty subset $S \subseteq J$ satisfies
\[
    \sum_{j \in S} p_j t_j < \sum_{i \in \Gamma(S)} m_i.
\]
Moreover, whenever Event~6 is executed, the same strict inequality holds for every nonempty subset $S \subseteq J$ at the current scaled prices.
\end{lemma}
\begin{proof}
Consider the start of a new phase. At this point, the algorithm has just computed a balanced flow $f$ of the full expenditure network $N(\pb,\mb,\tb)$, and the outer loop has not terminated, so the maximal surplus
\[
    \delta \coloneqq \max_{i \in B} \gamma(f)_i
\]
is strictly positive. Let $I = \argmax_{i \in B} \gamma(f)_i$ and $J = \Gamma(I)$.
We claim that no deleted edge from a good in $J$ to a buyer in $B\setminus I$ carries positive flow under $f$. Indeed, if some deleted edge $(g,i')$ with $g \in J$ and $i' \in B\setminus I$ carried positive flow, then there would exist a buyer $i \in I$ adjacent to $g$, and the restricted residual graph $R_N(f)$ would contain the path $i' \to g \to i$. But $\gamma(f)_{i'} < \delta = \gamma(f)_i$, contradicting \cref{lemma:balanced-flow-residual}. Hence deleting the edges from goods in $J$ to buyers in $B\setminus I$ leaves $f$ feasible, with the same buyer surpluses, and in particular every buyer in $I$ still has surplus $\delta$.

Let $S \subseteq J$ be any nonempty subset. Since all edges from goods in $J$ to buyers in $B\setminus I$ have been deleted, we have $\Gamma(S) \subseteq I$. Every source edge $(s,g)$ with $g \in S$ must be saturated by $f$: if $(s,g)$ had residual capacity, then choosing any buyer $i \in \Gamma(\{g\}) \subseteq I$ would give an augmenting path $s \to g \to i \to t$, since $g \to i$ always has residual capacity and $i \to t$ has residual capacity $\delta > 0$, contradicting that $f$ is a max flow. Hence each source edge into a good of $S$ is saturated, so the total flow sent out of $S$ equals $\sum_{j \in S} p_j t_j$. As all flow out of $S$ enters the buyers in $\Gamma(S)$, we obtain
\[
    \sum_{j \in S} p_j t_j
    \leq \sum_{i \in \Gamma(S)} (m_i - \delta)
    < \sum_{i \in \Gamma(S)} m_i.
\]

Now suppose Event~6 is executed in some inner iteration and let $\theta_6$ be the selected scaling factor. If the displayed strict inequality failed for some nonempty $S \subseteq J$ at the current scaled prices, then the Event~5 condition would hold for that same $S$ at some scaling factor at most $\theta_6$. Since Event~5 has priority over Event~6 whenever both occur at the same scaling factor, the algorithm would not execute Event~6. Hence these inequalities are strict when Event~6 fires.
\end{proof}

The following lemma gives a simple Hall-type characterization of when the source-side cut is a min-cut, which is used repeatedly throughout the remaining proofs.

\begin{lemma}
\label{lemma:source-cut-characterisation}
For any supply vector $\widetilde{\tb}$, the cut $(\{s\}, G \cup B \cup \{t\})$ is a min-cut of the network $N(\pb,\mb,\widetilde{\tb})$ if and only if
\[
    \sum_{j \in S} p_j \widetilde{t}_j \leq \sum_{i \in \Gamma(S)} m_i
    \qquad \text{for every } S \subseteq G.
\]
\end{lemma}
\begin{proof}
Let
\[
    C_s(\widetilde{\tb}) \coloneqq \sum_{j \in G} p_j \widetilde{t}_j
\]
denote the capacity of the cut $(\{s\}, G \cup B \cup \{t\})$ in $N(\pb,\mb,\widetilde{\tb})$. Every finite $s$-$t$ cut is determined by a set $S \subseteq G$ of goods placed on the source side: once $S$ is chosen, every buyer in $\Gamma(S)$ must also lie on the source side to avoid cutting an infinite-capacity edge. For this fixed set $S$, the minimum-capacity finite cut is therefore obtained by placing exactly the buyers in $\Gamma(S)$ on the source side, and its capacity is
\[
    C_s(\widetilde{\tb}) - \sum_{j \in S} p_j \widetilde{t}_j + \sum_{i \in \Gamma(S)} m_i.
\]
Hence the source-side cut is a min-cut if and only if its capacity $C_s(\widetilde{\tb})$ is at most the capacity of every such cut, which is equivalent to the stated inequalities.
\end{proof}

Recall the profitable supply vector $\bar{\tb}$ defined previously. For an inner-loop state with current sets $I$ and $J$, we also define the \emph{hybrid phase supply} vector $\tilde{\tb}^{(J)}$ by
\[
    \tilde t^{(J)}_j =
    \begin{cases}
        t_j & \text{if } j \in J,\\
        \bar t_j & \text{if } j \in G \setminus J.
    \end{cases}
\]
The following hybrid min-cut lemma is the key invariant maintained throughout the inner loop. It links the phase component (where actual supply is used) to the remainder (where profitable supply is used), and underpins the profitable-network invariant of \cref{lemma:profitable-mincut}.

\begin{lemma}
\label{lemma:hybrid-mincut}
At the start of each inner iteration, provided that the source-side cut is a min-cut of the profitable network $N(\pb,\mb,\bar{\tb})$ at the start of the enclosing phase, either Event~5 is already triggered at scaling factor $1$, or the source-side cut $(\{s\}, G \cup B \cup \{t\})$ is a min-cut of the current hybrid network $N(\pb,\mb,\tilde{\tb}^{(J)})$.
\end{lemma}
\begin{proof}
We argue by induction on inner iterations within a phase.

First we discuss the base case. At the start of the phase, the algorithm computes the sets $I$ and $J$, deletes the edges from goods in $J$ to buyers in $B\setminus I$, and enters the inner loop. Because $J=\Gamma(I)$ in the full expenditure network, no good in $G \setminus J$ is adjacent to a buyer in $I$. After the deletion step, no good in $J$ is adjacent to a buyer in $B\setminus I$ either. Hence the network $N(\pb,\mb,\tilde{\tb}^{(J)})$ decomposes into the disjoint union of the current phase component on $\{s\}\cup J \cup I \cup \{t\}$ and the remainder on $\{s\}\cup (G \setminus J)\cup(B\setminus I)\cup\{t\}$.

By \cref{lemma:no-tight-subset,lemma:source-cut-characterisation}, the strict slack at the start of the phase implies that the source-side cut is a min-cut of the phase component. On the remainder, the capacities are exactly the profitable capacities $p_j \bar t_j$, and the assumption on the profitable network implies that the source-side cut is a min-cut there as well. Therefore the source-side cut is a min-cut of the whole hybrid network at the start of the first inner iteration.

Now consider the inductive step. First suppose an inner iteration starts with Event~5 already triggered at scaling factor $1$. Since the algorithm executes the smallest-index event attaining the minimum scaling factor, Event~6 cannot be selected before Event~5 from this state: either $\theta_6>1$, or $\theta_6=1$ and Event~5 has priority. Thus the only events that can lead to another inner iteration before the pending Event~5 phase break are Events~2, 3, and~4, all at scaling factor $1$. If Event~1 or Event~5 is selected, the phase ends and there is no next inner iteration to check. If Event~3 or Event~4 is selected first at the same scaling factor, the phase set $J$, its source capacities, and the neighborhoods of subsets of $J$ are unchanged, so Event~5 is still triggered at the start of the next inner iteration. If Event~2 is selected first and the recomputed balanced flow is terminal, again there is no next inner iteration. Otherwise the restoration gives new sets $I^+$ and $J^+$. If some nonempty subset of $J^+$ is tight, Event~5 is triggered at scaling factor $1$ in the next inner iteration. If none is tight, \cref{lemma:source-cut-characterisation} gives the min-cut property on the new phase component. For the remainder, use the most recent certified hybrid network, namely the hybrid network immediately before this zero-scaling chain began. All intervening events in this pending-Event~5 case occur at scaling factor $1$, so prices and supplies have not changed; Events~2 and~3 only add mbpb edges, Event~4 removes only buyers isolated from the current network, and the new remainder goods are a subset of the preceding certified remainder. Hence the source-side min-cut inequalities still hold on the remainder, and the lemma holds in the next inner iteration as well.

It remains to treat an inner iteration that starts with the source-side cut being a min-cut of the current hybrid network and with Event~5 not already triggered at scaling factor $1$. During the scaling interval until the next event, the remainder component is unchanged. By \cref{lemma:source-cut-characterisation}, it is enough for the phase component to maintain
\[
    \sum_{j\in S} p_j t_j \leq \sum_{i\in \Gamma(S)} m_i
    \qquad \text{for every } S\subseteq J.
\]
During the scaling interval, all prices in $J$ are scaled by a common factor $\theta \geq 1$, so the left-hand sides $\sum_{j\in S} p_j t_j$ grow while the right-hand sides $\sum_{i\in \Gamma(S)} m_i$ are unchanged; any tightening would trigger Event~5. Moreover, Events~1, 5, and~6 all end the phase, so no further inner iteration is entered after any of them. It therefore suffices to check that Events~2, 3, and~4 preserve the hybrid min-cut property.

If Event~3 occurs, the algorithm only adds an edge within the remainder component, which can only enlarge neighborhoods and hence preserves the source-side min-cut inequalities there. If Event~4 occurs, the removed buyer lies in $Z$, so she is isolated and belongs to the remainder component; removing her does not change any neighborhood $\Gamma(S)$ for a subset of goods. Thus Events~3 and~4 preserve the hybrid min-cut property.

If Event~2 occurs and the recomputed balanced flow has zero surplus on every buyer, the outer loop terminates and no further inner iteration is entered.

If Event~2 occurs and the algorithm continues, then $I$ and $J$ are enlarged to new sets $I^+$ and $J^+$. By \cref{lemma:event2-restoration}, every arc deleted from goods in $J^+$ to buyers in $B\setminus I^+$ carries zero flow in the recomputed balanced flow, so the restored restricted network is feasible for the new phase component. Since $J^+ \supseteq J$, the new remainder goods satisfy $G \setminus J^+ \subseteq G \setminus J$, so the old hybrid min-cut inequalities already certify the new remainder. If the new phase component contains a nonempty set $S\subseteq J^+$ with
\[
    \sum_{j\in S} p_j t_j \geq \sum_{i\in \Gamma(S)} m_i,
\]
then Event~5 is triggered at scaling factor $1$ in the next inner iteration. Otherwise all such inequalities are strict, and \cref{lemma:source-cut-characterisation} shows that the source-side cut is a min-cut of the new hybrid network $N(\pb,\mb,\tilde{\tb}^{(J^+)})$.
\end{proof}

The following lemma is the global phase-to-phase invariant, established by induction over outer-loop iterations using \cref{lemma:hybrid-mincut}. It is the seller-side analogue of the buyer-side Hall condition, and together with \cref{lemma:source-cut-characterisation} it provides the seller-side Hoffman inequality needed in \cref{lemma:final-lb-feasible}.

\begin{lemma}
\label{lemma:profitable-mincut}
At the start of each outer loop iteration, and after the outer loop terminates and the full expenditure network is rebuilt, the source-side cut is a min-cut of the profitable network $N(\pb,\mb,\bar{\tb})$.
\end{lemma}
\begin{proof}
We argue by induction over the outer loop iterations.

At initialization, $\seg_j = 1$ and $p_j = d_{j1}$ for every good $j \in G$, so $\bar t_j = b_{j0} = 0$. Hence every source edge in the profitable network $N(\pb,\mb,\bar{\tb})$ has capacity $0$, and the source-side cut is trivially a min-cut.

Now assume that the claim holds at the start of some outer loop iteration. If the outer loop terminates immediately after the balanced-flow computation, then the current network is already the rebuilt full expenditure network used by the final processing step, so the final statement follows.

Otherwise, the algorithm enters the inner loop. By \cref{lemma:hybrid-mincut}, at the start of each inner iteration either Event~5 is already triggered at scaling factor $1$, or the source-side cut is a min-cut of the hybrid network $N(\pb,\mb,\tilde{\tb}^{(J)})$.

If Event~2 occurs and the recomputed balanced flow has zero surplus on every buyer, then the outer loop terminates from this inner iteration. If the hybrid min-cut alternative holds in this iteration, the only change made before the final processing step is the addition of a new mbpb edge, and adding edges can only enlarge neighborhoods; passing to profitable capacities and then rebuilding the final network therefore preserves the source-side min-cut property. If instead Event~5 was already triggered at scaling factor $1$, the zero-scaling argument in the next paragraph applies directly to the final profitable network.

At the end of the phase, first suppose the source-side cut is a min-cut of the final hybrid network. It remains to pass from the hybrid to the profitable network. If the phase ends with Event~5, the profitable network at phase end is obtained from the final hybrid network by possibly decreasing some source capacities on goods in $J$ that remain exactly at breakpoints; by \cref{lemma:source-cut-characterisation}, decreasing source capacities preserves the source-side min-cut property.

There remains the case in which the phase ends with Event~5 at scaling factor $1$ before the hybrid min-cut property is re-established. Let $N^c$ be the most recent certified hybrid network, namely the hybrid network immediately before the zero-scaling chain that leads to this Event~5. From $N^c$ to phase end, the prices and supplies have not changed. The intervening zero-scaling Events~2 and~3 only add mbpb edges, Event~4 removes only buyers isolated from the current network, and any good newly absorbed into the phase had profitable capacity in $N^c$ and has the same profitable capacity at phase end. The final profitable network is therefore obtained from a certified network by adding edges, deleting isolated buyers, and weakly decreasing source capacities. These operations preserve the source-side min-cut inequalities of \cref{lemma:source-cut-characterisation}.

If the phase ends with Event~6, we must additionally verify that the buyer-side change preserves the Hall inequalities $\sum_{j \in S} p_j \bar t_j \leq \sum_{i' \in \Gamma(S)} m_{i'}$ for every $S \subseteq G$. Let $R=G\setminus J$. Since $J=\Gamma(I)$ before the deletion step, no good in $R$ is adjacent to a buyer in $I$; and since arcs from $J$ to $B\setminus I$ were deleted, no good in $J$ is adjacent to a buyer outside $I$. Thus the neighborhoods of subsets of $J$ and subsets of $R$ are disjoint. The inequalities for subsets of $R$ are unaffected by Event~6, because the changed buyer belongs to $I$. Therefore it remains to prove the inequalities for subsets of $J$: adding the already-valid inequality for $S\cap R$ then gives the inequality for an arbitrary mixed set $S\subseteq G$. For $S \subseteq J$, recall $\Gamma(S) \subseteq I$. In $N_{I,J}^{(i,0)}$, a source-side cut with good side $S \subseteq J$ and buyer side $\Gamma(S)$ has capacity $\sum_{j \in J \setminus S} p_j t_j + \sum_{i' \in \Gamma(S) \setminus \{i\}} m_{i'}$, so being at least as large as the source-side cut capacity $\sum_{j \in J} p_j t_j$ is equivalent to $\sum_{i' \in \Gamma(S) \setminus \{i\}} m_{i'} \geq \sum_{j \in S} p_j t_j \geq \sum_{j \in S} p_j \bar t_j$.

First suppose the Event~6 step removes buyer $i$. Then the source-side cut \emph{is} a min-cut of $N_{I,J}^{(i,0)}$, so this inequality holds for every $S \subseteq J$, which is exactly what is needed after buyer $i$ is removed.

Now suppose the Event~6 step reduces buyer $i$'s budget to $m_i'$ instead. The maximal source-side min-cut $(S^*, T^*)$ has capacity
\[
    \sum_{j\in J}p_jt_j-m_i',
\]
where $m_i'=\sum_{j \in S^*} p_j t_j-\sum_{i'\in T^*\setminus\{i\}}m_{i'}$. For $S\subseteq J$ with $i\notin\Gamma(S)$ the inequality is unchanged. Now let $S\subseteq J$ with $i\in\Gamma(S)$. The source-side cut with good side $S$ and buyer side $\Gamma(S)$ has capacity
\[
    \sum_{j\in J\setminus S}p_jt_j+\sum_{i'\in\Gamma(S)\setminus\{i\}}m_{i'}.
\]
Since $(S^*,T^*)$ is a min-cut, this capacity is at least $\sum_{j\in J}p_jt_j-m_i'$, and therefore
\[
    \sum_{j\in S}p_jt_j
    \leq m_i'+\sum_{i'\in\Gamma(S)\setminus\{i\}}m_{i'}.
\]
Using $\bar t_j\leq t_j$ gives
\[
    \sum_{j\in S}p_j\bar t_j
    \leq m_i'+\sum_{i'\in\Gamma(S)\setminus\{i\}}m_{i'}
    =\sum_{i'\in\Gamma(S)}m_{i'}',
\]
which is the required Hall inequality.

If the phase ends with Event~1, let $H \subseteq J$ be the goods that hit their next breakpoints. Immediately before the segment update, each good $j\in H$ still has actual supply $t_j=b_{j,\seg_j}$, so its hybrid capacity is $p_j b_{j,\seg_j}$. Immediately after Event~1, we update $\seg_j$ to $\seg_j+1$, and the good now lies exactly at the new breakpoint; since $p_j = d_{j,\seg_j}$ after the update, \cref{def:profitable-supply} gives $\bar t_j = b_{j,\seg_j-1}$, so the profitable capacity becomes $p_j \bar t_j = p_j b_{j,\seg_j-1}$. Thus the profitable network after Event~1 is numerically identical to the pre-Event~1 hybrid network. Therefore, at the end of every phase, the source-side cut is a min-cut of the current profitable network.

The next outer-loop iteration begins by rebuilding the full expenditure network, i.e., by restoring all mbpb edges at the current prices. Adding edges can only enlarge each neighborhood $\Gamma(S)$, so by \cref{lemma:source-cut-characterisation} every source-side min-cut inequality that holds before the rebuild also holds afterwards. Therefore the source-side cut is again a min-cut of the full profitable network at the start of the next outer loop iteration, closing the induction. In particular, it is also a min-cut after the final rebuild when the outer loop terminates.
\end{proof}

The final and most important lemma establishes that the lower-bounded network in the algorithm's final processing step always admits a feasible flow. This is the central result used in \cref{prop:correctness}: it guarantees that the algorithm can always find an allocation satisfying both buyer clearing and seller profit maximization simultaneously.

\begin{lemma}
\label{lemma:final-lb-feasible}
Let $\pb$, $\mb$, $\tb$, and $\bar{\tb}$ be the final prices, budgets, actual supply, and profitable supply after the outer loop terminates. Let $\overline N$ be the lower-bounded network on the same vertices and arcs as $N(\pb,\mb,\tb)$ with bounds
\[
    \ell(s,j) = p_j \bar t_j,\quad u(s,j) = p_j t_j,\quad
    \ell(j,i) = 0,\quad u(j,i) = \infty,\quad
    \ell(i,t) = u(i,t) = m_i.
\]
Then $\overline N$ has a feasible flow.
\end{lemma}
\begin{proof}
We reduce the lower-bounded flow problem to a circulation problem in the standard way. Add an auxiliary arc $(t,s)$ of infinite capacity and impose the lower and upper bounds stated above on all original arcs of $\overline N$. A feasible circulation in this auxiliary network restricts to a feasible flow in $\overline N$, and conversely any feasible flow in $\overline N$ can be completed to such a circulation by sending its value on the auxiliary arc. We now verify Hoffman's circulation criterion \citep{Sch-book} for this auxiliary network, which implies the existence of a feasible flow.
For a set of vertices $X$, Hoffman's inequality is
\[
    \ell(\delta^-(X)) \leq u(\delta^+(X)),
\]
where $\delta^-(X)$ and $\delta^+(X)$ are the arcs entering and leaving $X$. We show that every nontrivial finite inequality is implied by one of two Hall-type conditions.

Firstly, the outer loop terminates only when a balanced flow of the rebuilt full network $N(\pb,\mb,\tb)$ has zero surplus on every buyer. Rebuilding only restores mbpb edges, so this zero-surplus flow remains feasible in $N(\pb,\mb,\tb)$. Hence the full network admits a flow that saturates every buyer arc $(i,t)$. Therefore, for every subset $T \subseteq B$ of buyers,
\begin{equation}
\label{eq:buyer-Hall-condition}
    \sum_{i \in T} m_i \leq \sum_{j \in \Gamma(T)} p_j t_j,
\end{equation}
because the buyers in $T$ can receive flow only from their neighboring goods $\Gamma(T)$. This is the buyer-side condition: the upper capacities of the source arcs incident to $\Gamma(T)$ can supply the exact demand of $T$.

Secondly, \cref{lemma:profitable-mincut,lemma:source-cut-characterisation} imply that
\begin{equation}
    \label{eq:seller-Hall-condition}
    \sum_{j \in S} p_j \bar t_j \leq \sum_{i \in \Gamma(S)} m_i
    \qquad \text{for every } S \subseteq G.
\end{equation}
This is the seller-side condition: the mandatory lower-bound flow out of any set of goods can be absorbed by adjacent buyers.

All that remains is to check that these two conditions cover Hoffman's inequalities. Let $X$ be any vertex set. If a goods-to-buyers arc leaves $X$, or if $t \in X$ and $s \notin X$, then $u(\delta^+(X))=\infty$, so the inequality is automatic. If $s \in X$ and $t \notin X$, the inequality is also automatic, since no positive-lower-bound arc enters $X$.

The only remaining cases are therefore the following. Suppose that $s,t \notin X$, and let $S=X\cap G$. If the outgoing capacity is infinite, we are done. So assume finite outgoing capacity, which implies $\Gamma(S)\subseteq X\cap B$. So Hoffman's inequality is implied by \eqref{eq:seller-Hall-condition}, as
\[
    \ell(\delta^-(X))=\sum_{j\in S}p_j\bar t_j
    \leq \sum_{i\in \Gamma(S)} m_i
    \leq u(\delta^+(X)).
\]
Next, suppose that $s,t \in X$, and let $T=B\setminus X$. Finite outgoing capacity implies every neighbor of a buyer in $T$ lies in $G\setminus X$, so $\Gamma(T)\subseteq G\setminus X$, and Hoffman's inequality is implied by \eqref{eq:buyer-Hall-condition}, as
\[
    \ell(\delta^-(X))=\sum_{i\in T}m_i
    \leq \sum_{j\in \Gamma(T)}p_jt_j
    \leq u(\delta^+(X)).
\]
Thus all Hoffman's inequalities hold, and the auxiliary network has a feasible circulation. Removing the auxiliary arc $(t,s)$ gives a feasible flow in $\overline N$.
\end{proof}

\section{Correctness of the Algorithm}
We now see that the algorithm returns a competitive equilibrium when it terminates. By \cref{lemma:buyer-demand}, a buyer with $\alpha_i < 1$ demands only the empty bundle; a buyer with $\alpha_i = 1$ demands any affordable bundle consisting solely of mbpb goods, including the empty bundle; and a buyer with $\alpha_i > 1$ demands an affordable bundle consisting solely of mbpb goods on which she spends her entire budget. Note that when $\alpha_i = 1$ the empty bundle is always demanded, so a buyer whose algorithmic budget has been reduced by an Event~6b refund (to a value at most her original budget) is still satisfied by the empty bundle, or by any affordable bundle of mbpb goods up to her reduced budget. The auctioneer, meanwhile, maximizes his profit by selling a quantity of each good $j$ as upper- and lower-bounded in \cref{lemma:profit-maximising}. We now prove that the outcome $(\pb, \xb)$ returned by the algorithm is a competitive equilibrium; that is, every buyer receives a bundle $\xb_i$ she demands at $\pb$, and the aggregate allocation $\sum_{i \in B} \xb_i$ maximizes the auctioneer's profit over all possible bundles $\zb \in \R^G_{\geq 0}$ at $\pb$.

\begin{proposition}
\label{prop:correctness}
The algorithm returns a competitive equilibrium when it terminates.
\end{proposition}
\begin{proof}
Let $\pb$, $\mb$, $\tb$, $(\seg_j)_{j \in G}$ be the final prices, budgets, supply, and segments. Let $\bar{\tb}$ be the final profitable supply vector and let $\overline N$ be the lower-bounded network from the final processing step. By \cref{lemma:final-lb-feasible}, $\overline N$ admits a feasible flow $f$. The allocation $\xb$ returned is defined by $x_{ij} = \frac{f(j,i)}{p_j}$ for each arc $(j,i)$ in $\overline N$, and $x_{ij} = 0$ otherwise (including all buyers removed by the algorithm).

We now verify that every buyer demands her allocated bundle at final prices $\pb$, by checking the three cases of \cref{lemma:buyer-demand}. Since prices increase monotonically, every buyer's bang per buck $\alpha_i$ decreases monotonically throughout the algorithm. We use two facts throughout: the arcs of $\overline N$ assign only mbpb goods, and every buyer who remains in the market spends her entire final budget because $\overline N$ enforces $f(i,t)=m_i$.

\textbf{Case 1:} Suppose a buyer has mbpb $\alpha_i < 1$ at $\pb$. By construction, such a buyer must have been removed, either during initialization or when she reached $\alpha_i=1$ in an Event~4 or Event~6 iteration. Hence she receives the empty bundle, which she demands by \cref{lemma:buyer-demand}.

\textbf{Case 2:} Suppose a buyer has mbpb $\alpha_i = 1$ at $\pb$. Such a buyer may have been removed in Event~4 or in the removal branch of Event~6, in which case she receives the empty bundle. Otherwise, she remains in the market, possibly with a reduced budget due to the budget-reduction branch of Event~6. In that case she receives only mbpb goods and spends her entire final budget. Either allocation is demanded when $\alpha_i=1$, by \cref{lemma:buyer-demand}.

\textbf{Case 3:} Suppose a buyer has mbpb $\alpha_i > 1$ at $\pb$. Such a buyer was never removed and was never involved in Event~6, which fires only when $\alpha_i=1$. Hence her final budget equals her original budget. She receives only mbpb goods and spends her entire original budget, as required by \cref{lemma:buyer-demand}.

Finally, the feasible flow in $\overline{N}$ satisfies $\bar{t}_j \leq y_j \leq t_j$ for every good $j$, where $y_j \coloneqq \sum_{i \in \widehat{B}} x_{ij}$. By definition of $\bar{\tb}$ and $\tb$: if $p_j = d_{j,\seg_j}$ then $\bar{t}_j = b_{j,\seg_j-1}$ and $t_j = b_{j,\seg_j}$, so $y_j \in [b_{j,\seg_j-1}, b_{j,\seg_j}]$; and if $p_j \in (d_{j,\seg_j}, d_{j,\seg_j+1})$ then $\bar{t}_j = t_j = b_{j,\seg_j}$, so $y_j = b_{j,\seg_j}$. By \cref{lemma:profit-maximising}, these are exactly the profit-maximizing quantities for the auctioneer at prices $\pb$. Hence the returned outcome is a competitive equilibrium.
\end{proof}

\section{The Running Time}
Together with \cref{prop:correctness}, the following theorem completes the proof that the algorithm computes a competitive equilibrium in polynomial time.

\begin{theorem}
\label{thm:running-time}
Suppose $\mathcal{A} = (B, G, u, \mb, \bb, \db)$ is an Arctic auction instance with rational utilities, budgets, and marginal cost steps. Let $n = |G|$, $m = |B|$, $K = \sum_{j \in G} K_j$, and let $L$ be the maximum bit length of the numerator or denominator of any input rational appearing in $\ub$, $\mb$, $\bb$, or $\db$. Then \Cref{alg:arctic-with-costs} solves $\mathcal{A}$ with at most
\[
O( (m+K)(m+n)^5(L + \log(m+n)) )
\]
outer loop iterations, and performs at most
\[
    O( (m+K)mn(n+\log(m+n))(m+n)^5(L + \log(m+n)) )
\]
max flow computations.

\end{theorem}

The bound on our algorithm's running time is likely to be pessimistic; we optimize for a clean polynomial rather than sharp exponents.

Our running-time analysis starts from the squared-surplus potential used by \citet{arctic_markets_production, DPSV}, which maps any network flow $f$ and budgets $\mb$ to the sum of the squared surpluses of all buyers:
\[
\Phi(f, \mb) = \sum_{i \in B} (m_i - f(i,t))^2.
\]
Recall that a phase is an iteration of the outer loop together with all its inner iterations. A phase is \emph{nonterminal} if its last inner iteration enters Event~$1$, $5$, or $6$; it is a \emph{terminal phase} if the outer loop terminates during it, either because a balanced-flow computation at the start of the outer-loop iteration yields zero surplus on every buyer, or because an Event~$2$ inner iteration yields zero surplus on every buyer. In the latter case the terminal phase has no closing Event~$1$, $5$, or $6$, and is charged as a single additional phase for the purposes of \cref{prop:phase-number-bound}. We subdivide Event~6 into Event~6a, in which a buyer is removed, and Event~$6b$, in which we retain the buyer and only partially return her budget. This leads to the following four types of nonterminal phases:

\begin{description}
    \item[Type 1:] A phase that ends with an inner iteration entering Event~1, during which the supply of some good is increased.
    \item[Type 5:] A phase that ends with an inner iteration entering Event~5 during which a nonempty subset of goods goes tight.
    \item[Type 6a:] A phase that ends with an inner iteration entering Event~6 during which a buyer is removed.
    \item[Type 6b:] A phase that ends with an inner iteration entering Event~6 during which a buyer is not removed but her budget is reduced.
\end{description}

In the fixed-supply Arctic auction of \citet{arctic_markets_production}, source capacities are fixed multiples of current prices and buyer budgets remain fixed until exit. Here, source capacities are products of prices and endogenous supplies, supplies may jump at marginal-cost breakpoints, and Event~$6b$ can partially refund a buyer rather than remove her. Accordingly, the bit-complexity step tracks new rationality certificates, while the potential-decrease step handles the new event types.

The proof has four parts.

\subsection{Checkpoint Bit Complexity}
\textbf{Step 1: potential checkpoints.}
We first lower-bound every nonzero value of the potential at the balanced-flow checkpoints where the algorithm evaluates it. This requires a new certificate of rationality for the current prices, source capacities, and active budgets: in addition to mbpb equations and breakpoint anchors, the certificate must keep historical tight-set equations and the cut equations created by partial refunds.

Let $m_i^{\mathrm{init}}$ denote buyer $i$'s initial budget, let $M = \sum_{i \in B} m_i^{\mathrm{init}}$ denote the total budget of all buyers at the beginning of the algorithm, let $h=m+n$, and let $L$ be the maximum bit length of the numerator or denominator of any input rational among $\ub$, $\mb$, $\bb$, and $\db$.

Since current budgets decrease monotonically over the course of the algorithm, every feasible flow $f$ satisfies
\[
0 \leq f(i,t) \leq m_i \leq m_i^{\mathrm{init}}
\qquad \text{for every } i \in B.
\]
Hence, we get
\[
0 \leq \Phi(f, \mb) \leq \sum_{i \in B} (m_i^{\mathrm{init}})^2 \leq M^2.
\]

We next record the bit-complexity invariant used to lower-bound nonzero potential values. A \emph{balanced-flow checkpoint} is the algorithmic state immediately after one of the balanced-flow computations performed by the algorithm, either at the start of an outer-loop iteration or after an Event~$2$ update. These are exactly the states at which the running-time proof evaluates $\Phi$.

The bookkeeping is twofold: source capacities are the money values $\sigma_j=p_jt_j$ of the current supplies, and Event~$6b$ can replace an original buyer-budget equation by a cut equation fixing the buyer's current budget. In the proof it is cleaner to certify prices first; the desired bound for $\sigma_j$ then follows because every current supply $t_j$ is an input breakpoint.

\begin{lemma}
\label{lemma:checkpoint-bit-complexity}
At every balanced-flow checkpoint, each current source capacity $\sigma_j \coloneqq p_jt_j$ and each current active buyer budget $m_i$ has numerator and denominator bit length
\[
    O((m+n)^2(L+\log(m+n))).
\]
Moreover, all current source capacities and active buyer budgets can be written with a common denominator of the same asymptotic bit length.
\end{lemma}
\begin{proof}
The strategy is to certify the current checkpoint by a small full-rank linear system. We certify prices and active budgets, using auxiliary variables $\lambda_i=1/\alpha_i$ for the active buyers, and recover source capacities only at the end from $\sigma_j=p_jt_j$.

At a checkpoint, the variables are the current prices $p_j$, the current active budgets $m_i$, and the variables $\lambda_i$. The certificate is a maximal independent subsystem of equations of the following forms:
\begin{itemize}
    \item for each edge $(j,i)$ of a spanning forest of the current mbpb graph, the equation
    \[
        p_j=u_{ij}\lambda_i;
    \]
    \item anchor equations $\lambda_i=1$ for buyers that have become indifferent, and $p_j=d_{j,\seg_j}$ for goods pinned at a marginal-cost breakpoint;
    \item tight-set equations
    \[
        \sum_{j\in S}t^S_jp_j=\sum_{i\in T^S}m_i;
    \]
    \item for each buyer, one budget equation: either $m_i=m_i^{\mathrm{init}}$ or the most recent Event~$6b$ cut equation
    \[
        m_i=\sum_{j\in S_i}t^{S_i}_jp_j-\sum_{i'\in T_i\setminus\{i\}}m_{i'} .
    \]
\end{itemize}
Here $t^S_j$ denotes the supply of good $j$ at the time a tight-set equation for $S$ was revealed, and $T^S$ denotes the buyer set on the right-hand side of that equation. Similarly, $t^{S_i}_j$ and $T_i$ refer to the supply values and buyer set in the most recent Event~$6b$ cut equation for buyer $i$. Thus every such supply coefficient is one of the input breakpoints $b_{jk}$. These are historical certificate equations: when such an equation is retained, its coefficients remain the supply values from the time of revelation. Later Event~$1$ updates may change the current supply of a good, but they do not change prices or budgets, so any retained historical equation remains a valid linear equation for the current price and budget variables.

We maintain the invariant that, after redundant equations are discarded, the retained subsystem has the current values of $p_j$, $m_i$, and $\lambda_i$ as its unique solution. Initially this is immediate from the anchors $p_j=d_{j1}$, the original-budget equations, and one mbpb equation for each active buyer. Now suppose the invariant holds at the preceding checkpoint. During one execution of the inner loop, before the next event is reached, the retained equations continue to determine all variables except for the single common scaling factor of the current phase component. The event reached by this scaling step supplies exactly the equation that fixes this remaining freedom. Event~$2$ adds a newly tight mbpb equation for an edge from a good outside the old phase set to a buyer inside it. This edge need not be a forest edge in the full mbpb graph after the restoration; if it lies on a cycle, we simply discard redundant mbpb equations and keep a maximal independent forest subsystem. What matters is that the newly tight equality is not implied by the retained subsystem before the scaling step: if it were, it would already have held before the step began, contradicting the definition of Event~$2$. Event~$5$ adds the tight-set equation for a set that has just become tight, and Event~$6$ adds the anchor $\lambda_i=1$. The same independence argument applies to these equations: if the corresponding tight-set equation or indifference condition were already implied, the event would have held before this scaling step began. In Event~$6b$, the cut equation above then replaces the previous budget equation for buyer $i$; its coefficient on $m_i$ is $1$, so this is a rank-preserving pivot that fixes the buyer's new budget. Event~$6a$ only deletes a buyer and the equations used to determine her variables, so it preserves uniqueness for the remaining variables.

It remains to check Event~$1$. Event~$1$ changes supplies but not prices or budgets. If $j$ reaches the next marginal-cost breakpoint, then after the segment update
\[
    p_j=d_{j,\seg_j}.
\]
We therefore add this price anchor to the certificate, discarding redundant equations if necessary. Since the current values of $p_j$, $m_i$, and $\lambda_i$ do not change during the supply update, all previously retained historical equations remain valid. Thus the full-rank certificate invariant is preserved across Event~$1$ as well.

All coefficients in the certificate are $0$, $\pm1$, or products of at most two input rationals from $\ub$, $\bb$, and $\db$; the right-hand sides are input rationals or zero. Now clear denominators row by row in the retained subsystem. A row may contain $O(m+n)$ rational coefficients, so after clearing denominators the entries have bit length $O((m+n)L)$. The resulting integer system has at most $2m+n$ equations. Applying Cramer's rule and Hadamard's inequality \citep{Sch-book}, we see that every certified price $p_j$ and active budget $m_i$ has numerator and denominator bit length
\[
    O((m+n)^2(L+\log(m+n))).
\]
More precisely, the determinant of the row-cleared coefficient matrix is a common denominator for all certified prices and active budgets, and its bit length satisfies the same bound. Finally, $\sigma_j=p_jt_j$ and the current supply $t_j=b_{j,\seg_j}$ is an input rational of bit length at most $L$. Multiplying by $t_j$ increases individual numerator and denominator bit length by at most $O(L)$. Taking the product of the input denominators of the at most $n$ current supplies and multiplying it by the determinant above gives a common denominator for all current source capacities and active budgets, still with bit length $O((m+n)^2(L+\log(m+n)))$.
\end{proof}

Next, we derive a lower bound on the \emph{non-zero} values that $\Phi$ can take at the balanced-flow checkpoints where the algorithm evaluates the potential.
\begin{lemma}
\label{lemma:Phi-lower-bound}
At every state immediately after a balanced-flow computation, either $\Phi = 0$ or
\[
   \Phi \geq 2^{-O((m+n)^2(L + \log(m+n)))}.
\]
\end{lemma}
\begin{proof}
We only need to consider states immediately after a balanced-flow computation, because these are exactly the states at which our running time analysis evaluates $\Phi$. By \cref{lemma:checkpoint-bit-complexity}, all source capacities and active buyer budgets at such a state can be written with a common denominator of bit length $O((m+n)^2(L+\log(m+n)))$.

Now consider the balanced-flow computation itself. After multiplying all source capacities and buyer budgets by a common denominator
\[
    Q \le 2^{O((m+n)^2(L + \log(m+n)))},
\]
the corresponding max-flow instances used by the balanced-flow subroutine have integral capacities. The DPSV balanced-flow algorithm partitions buyers into equal-surplus blocks and computes a max flow on each block's subnetwork. Since those subnetworks have integral capacities after scaling by $Q$, the max flow on each block can be chosen to be integral (after scaling), so the total flow into each block is an integer multiple of $1/Q$. If buyer $i$ lies in a block of size $k$, then her surplus equals the block's total budget minus its total flow, divided by $k$. Thus $\gamma_i$ is an integer multiple of $1/(Qk)$ for some $k \leq m$.

If $\Phi > 0$, then at least one buyer has non-zero surplus, say $\gamma_i \neq 0$, and therefore $|\gamma_i| \ge 1/(Qm)$. It follows that
\[
    \Phi = \sum_{i \in B} \gamma_i^2 \ge \gamma_i^2 \ge \frac{1}{(Qm)^2}
    = 2^{-O((m+n)^2(L + \log(m+n)))},
\]
where the last step uses $\log(Qm) = O((m+n)^2(L+\log(m+n)))$ and $\log m = O(\log(m+n))$.
\end{proof}

\subsection{Potential Decrease Across Phases}
\label{section:bounding-phases}
\textbf{Step 2: reset and decrease phases.}
The phases that do not fit a monotone potential-decrease argument are controlled directly: a Type~1 phase advances a supply segment, and a Type~$6a$ phase removes a buyer.

\begin{observation}
\label{obs:finite-phases}
There are at most $\sum_{j \in G} K_j$ phases of Type $1$ (which increases the supply of a good), and at most $m$ phases of Type $6a$ (which removes a buyer).
\end{observation}

All remaining nonterminal phases end with Event~5 or Event~$6b$. Their analysis is local to intervals in which the actual supplies on the active phase are fixed. On such intervals the balanced-flow surplus-reduction mechanism from \citet{arctic_markets_production} can still be used, but only after isolating the restricted phase network and checking that the extra operations of our algorithm preserve the relevant comparison. Event~2 restores the phase component by adding residually reachable buyers and deleting only zero-flow arcs; Event~$6b$ changes the budget vector and must be compared through an explicit post-refund feasible flow. The next three lemmas make these adaptations precise.

In the Type~5 and Type~$6b$ proofs below, we compare balanced-flow checkpoints with unchanged buyer budgets. Between two such checkpoints, the later network is obtained from the earlier one by increasing source capacities on currently scaled goods and adding mbpb arcs. After the later balanced flow is computed, the algorithm may delete arcs, but only arcs carrying zero flow in that later balanced flow. Events~3 and~4 are accounted for separately: Event~3 only adds an mbpb arc incident to an isolated buyer outside the phase component, while Event~4 removes an isolated buyer, leaving the flow on all remaining buyers and goods unchanged and deleting that buyer's nonnegative contribution to $\Phi$.

\begin{lemma}
\label{lemma:one-iteration-surplus-drop}
Let $f$ and $f'$ be earlier and later balanced flows computed within a Type~5 or Type~$6b$ phase, before any Event~$6b$ refund. Then $\Phi(f',\mb)\leq \Phi(f,\mb)$. Moreover, if a buyer $i$ present at both computations satisfies $\gamma(f')_i=\gamma(f)_i-\sigma$ for some $\sigma>0$, then $\Phi(f',\mb)\leq \Phi(f,\mb)-\sigma^2$.
\end{lemma}
\begin{proof}
Along this part of the phase, the only possible buyer-budget change would be an Event~$6b$ refund, which is excluded. Event~4 may remove buyers, but each such buyer lies in $Z$ and is isolated, so she receives zero flow. Removing her and her incident zero-flow arcs deletes the nonnegative term $m_i^2$ from the potential and does not affect feasibility of the earlier flow on the remaining network. Apart from these harmless removals, the later network is obtained from the earlier one by weakly increasing source capacities on scaled goods and adding mbpb arcs; Event~3 is of this form, since it only adds an mbpb arc incident to a buyer in $Z$. Hence the earlier balanced flow, restricted to the buyers still present, remains feasible in the later network before any zero-flow deletion following $f'$. The standard DPSV surplus-drop lemma \citep[Lemma~8.3]{DPSV} therefore applies to the two balanced flows on the common remaining buyers. It gives both monotonicity of the squared surplus norm and the stated $\sigma^2$ improvement when a buyer present at both computations has surplus decrease $\sigma$. If zero-flow arcs are deleted after $f'$ is computed, then $f'$ remains feasible in the restored network by construction.
\end{proof}

\begin{lemma}
\label{lemma:phase-5-reduction}
During a phase of Type~5 (in which a nonempty set goes tight), the potential $\Phi$ is reduced to at most a $(1 - \frac{1}{h^3})$ fraction of its previous value.
\end{lemma}
\begin{proof}
Let $f_0$ be the balanced flow at the start of a Type~5 phase, and let $\delta=\max_{i\in B}\gamma(f_0)_i$. By definition, a Type~5 phase is a nonterminal phase ending with Event~5; hence the initial balanced flow did not have zero surplus on every buyer, and $\delta>0$. Let $I_0=\argmax_{i\in B}\gamma(f_0)_i$ and $J_0=\Gamma(I_0)$.

During the phase, the algorithm scales the goods in the current set $J$. If Event~2 occurs, \cref{lemma:event2-restoration} shows that the restored phase component has zero-flow deletions and that $J$ strictly grows. Let $(f_0,I_0,J_0),(f_1,I_1,J_1),\ldots,(f_k,I_k,J_k)$ be the balanced-flow checkpoints in this phase, where $f_k$ is the checkpoint immediately before the terminal scaling interval. Event~2 strictly enlarges $J$, so $k\leq n$.

Let $N^+$ be the restricted network at the phase-ending prices, just before the Event~5 phase break, and let $f^+$ be a balanced flow of $N^+$ with the same budgets. At this moment some nonempty $S\subseteq J_k$ satisfies
\[
    \sum_{j\in S}p_jt_j\geq \sum_{i\in \Gamma(S)}m_i,
\]
where $\Gamma(S)\subseteq I_k$ because all arcs from goods in $J_k$ to buyers outside $I_k$ have been deleted. Since every good in $J_k$ is adjacent to a buyer in $I_k$, the set $\Gamma(S)$ is nonempty. We claim that at least one buyer in $\Gamma(S)$ has zero surplus under $f^+$. Otherwise every buyer in $\Gamma(S)$ has residual capacity on her sink arc. If some source arc $(s,g)$ with $g\in S$ had residual capacity, then choosing any buyer $i\in\Gamma(\{g\})$ would give an augmenting path $s\to g\to i\to t$, contradicting maximality of $f^+$. Hence all source arcs into $S$ are saturated. All flow leaving $S$ enters buyers in $\Gamma(S)$, so the total inflow to $\Gamma(S)$ is at least $\sum_{j\in S}p_jt_j\geq\sum_{i\in\Gamma(S)}m_i$, contradicting the assumption that every buyer in $\Gamma(S)$ has positive surplus. Thus the minimum surplus on the current phase set falls from $\delta$ at $f_0$ to $0$ at $f^+$ over at most $k+1\leq n+1\leq h$ same-budget intervals.

Therefore some interval has a drop of at least $\delta/h$ in the minimum surplus of the current phase set. If the buyer attaining the lower minimum at the end of such an interval was newly added by an Event~2 restoration, then she has a residual path to the previous phase set. By \cref{lemma:balanced-flow-residual}, either her surplus is at least the surplus of a buyer in the previous phase set at the same balanced flow, or that previous buyer has already dropped by at least as much. Hence some buyer already present at the start of the interval has surplus decrease at least $\delta/h$. Applying \cref{lemma:one-iteration-surplus-drop} to that interval, and using monotonicity from the same lemma on all other same-budget intervals, gives $\Phi(f^+,\mb) \leq \Phi_{\mathrm{before}}-\delta^2/h^2$.

The next outer-loop checkpoint is computed after rebuilding the full expenditure network, which only restores mbpb arcs. Applying \cref{lemma:one-iteration-surplus-drop} once more gives $\Phi_{\mathrm{after}}\leq \Phi(f^+,\mb)$. Since every buyer surplus is at most $\delta$ at the beginning of the phase and there are at most $m\leq h$ buyers, $\Phi_{\mathrm{before}}\leq h\delta^2$. Combining gives
\[
    \Phi_{\mathrm{after}} \leq \Phi_{\mathrm{before}}\!\left(1-\frac{1}{h^3}\right).
\]
\end{proof}

\begin{lemma}
\label{lemma:phase-6b-reduction}
During a phase of Type~$6b$ (in which a buyer is not removed but refunded part of her budget), the potential $\Phi$ is reduced to at most a $(1 - \frac{1}{4h^3})$ fraction of its previous value.
\end{lemma}
\begin{proof}
Use the same notation as in the previous proof. Let $\delta$ be the maximum buyer surplus at the start of the Type~$6b$ phase. The phase ends when a buyer $i\in I$ reaches $\alpha_i=1$ and the Event~6 min-cut computation reduces, but does not eliminate, her remaining budget. Let $\rho>0$ be the amount by which $i$'s budget is reduced, and write $m_i'=m_i-\rho$ for the post-refund budget. Let $f_k$ be the last balanced-flow checkpoint before the refund. Let $N^-$ be the same-budget network at the Event~6b prices, immediately before buyer $i$'s budget is reduced, and let $f^-$ be a balanced flow of $N^-$.

We will use the following elementary fact: in an expenditure network, if a feasible flow $g$ can be extended to a maximum flow, then it can be extended to one whose potential, computed with the same buyer budgets, is at most that of $g$. Indeed, augmenting along a residual $s$-$t$ path only increases the sink flow of the path's final buyer and leaves every other buyer's sink flow unchanged. Thus buyer surpluses only weakly decrease componentwise. Consequently, the balanced flow has potential no larger than the potential of any feasible flow in the same network.

We first record a comparison between $f^-$ and the post-refund instance. If $f^-(i,t)\leq m_i'$, then $f^-$ is feasible after the refund and buyer $i$'s surplus weakly decreases. If $f^-(i,t)>m_i'$, let $\eta=f^-(i,t)-m_i'$. Since flow conservation at $i$ gives $\sum_j f^-(j,i)=f^-(i,t)$, choose numbers $0\leq r_j\leq f^-(j,i)$ over incoming arcs $(j,i)$ whose sum is $\eta$. Reducing $f^-(i,t)$ by $\eta$, and reducing both $f^-(j,i)$ and $f^-(s,j)$ by $r_j$ for each such good $j$, gives a feasible post-refund flow: source feasibility follows from $f^-(s,j)\geq f^-(j,i)\geq r_j$. Buyer $i$ then has surplus zero, and every other buyer has the same surplus as under $f^-$. Thus, in all cases, the post-refund network has a feasible flow whose post-refund potential is at most $\Phi(f^-,\mb)$; by the preceding paragraph and balanced-flow optimality,
\begin{equation}
\label{eq:post-refund-comparison}
    \Phi_{\mathrm{after}}\leq \Phi(f^-,\mb).
\end{equation}

First suppose that, before this final refund, the minimum surplus on the current phase set falls to at most $\delta/2$ at one of the same-budget balanced flows up to and including $f^-$. The same averaging argument as in the Type~$5$ proof gives an interval in which an already-present phase buyer's surplus decreases by at least $\delta/(2h)$. By \cref{lemma:one-iteration-surplus-drop}, monotonicity on the remaining same-budget intervals, \eqref{eq:post-refund-comparison}, and $\Phi_{\mathrm{before}}\leq h\delta^2$,
\[
    \Phi_{\mathrm{after}}
    \leq \Phi_{\mathrm{before}}-\frac{\delta^2}{4h^2}
    \leq \Phi_{\mathrm{before}}\!\left(1-\frac{1}{4h^3}\right).
\]
This proves the claimed bound in this case.

It remains to consider the case in which every buyer in the current phase set has surplus greater than $\delta/2$ at all these same-budget balanced flows. The fixed-supply Event~6b analysis of \citet[Lemma~8]{arctic_markets_production}, applied to our restricted phase network $N_{I,J}^{(i,0)}$, shows that the sum of the pre-refund drop in buyer $i$'s surplus and the refund is greater than $\delta/2$. Write
\[
    \Delta=\gamma(f_k)_i-\gamma(f^-)_i .
\]
Then either $\Delta\geq\delta/4$ or $\rho\geq\delta/4$.

First suppose $\Delta\geq\delta/4$. By \cref{lemma:one-iteration-surplus-drop},
\[
    \Phi(f^-,\mb)\leq \Phi(f_k,\mb)-\delta^2/16
    \leq \Phi_{\mathrm{before}}-\delta^2/16 .
\]
Together with \eqref{eq:post-refund-comparison}, this gives $\Phi_{\mathrm{after}}\leq\Phi_{\mathrm{before}}-\delta^2/16$.

It remains to consider the case $\rho\geq\delta/4$. Let $f_{\mathrm{pre}}$ be the feasible flow from the fixed-supply analysis that saturates buyer $i$'s arc ($f_{\mathrm{pre}}(i,t)=m_i$) and leaves all other surpluses equal to $\gamma(f_k)_{i''}$. Since Event~$6b$ does not remove the buyer, $m_i'>0$, and hence $\rho<m_i=f_{\mathrm{pre}}(i,t)$. Choose numbers $0\leq r_j\leq f_{\mathrm{pre}}(j,i)$ over incoming arcs $(j,i)$ whose sum is $\rho$; this is possible because flow conservation at $i$ gives $\sum_j f_{\mathrm{pre}}(j,i)=m_i$. Form $\tilde f$ by reducing $f_{\mathrm{pre}}(i,t)$ by $\rho$, and by reducing both $f_{\mathrm{pre}}(j,i)$ and $f_{\mathrm{pre}}(s,j)$ by $r_j$ for each such good $j$. The source reductions are feasible because flow conservation at each good gives $f_{\mathrm{pre}}(s,j)\geq f_{\mathrm{pre}}(j,i)\geq r_j$. This preserves flow conservation at every buyer and good, respects all capacities, and gives a feasible flow for $N(\pb,\mb',\tb)$ with $\gamma(\tilde f)_i=0$ and $\gamma(\tilde f)_{i''}=\gamma(f_k)_{i''}$ for all $i''\neq i$. Since buyer $i$ belongs to the current phase set and $\gamma(f_k)_i>\delta/2$, the elementary fact above applied to $\tilde f$ gives
\[
    \Phi_{\mathrm{after}}=\Phi(f^*,\mb')
    \leq\Phi(\tilde f,\mb')
    =\Phi(f_k,\mb)-\gamma(f_k)_i^2
    \leq\Phi_{\mathrm{before}}-\tfrac{\delta^2}{4}
    \leq\Phi_{\mathrm{before}}-\tfrac{\delta^2}{16}.
\]

In either subcase of this remaining case,
\[
    \Phi_{\mathrm{after}}
    \leq \Phi_{\mathrm{before}}-\left(\frac{\delta}{4}\right)^2
    \leq \Phi_{\mathrm{before}}\!\left(1-\frac{1}{16h}\right)
    \leq \Phi_{\mathrm{before}}\!\left(1-\frac{1}{4h^3}\right),
\]
because $h\geq 2$.
\end{proof}

\textbf{Step 3: counting phases.}
The lower bound from Step~1 and the multiplicative decrease from Step~2 bound the number of Type~5 and Type~$6b$ phases between two reset phases. Since the reset phases themselves are bounded by \cref{obs:finite-phases}, we obtain the total number of phases.

\begin{proposition}
\label{prop:phase-number-bound}
The number of phases in the algorithm is bounded by $O((m+K)(m+n)^5(L + \log(m+n)))$.
\end{proposition}
\begin{proof}
Suppose there are $R$ consecutive phases of Type $5$ or $6b$, and let $\Phi_\text{before}$ and $\Phi_\text{after}$ denote the value of the potential function at the balanced-flow checkpoints immediately before and immediately after this block of phases. At the beginning of such a block, either the algorithm has just started or the preceding phase ended with Event~$1$ or Event~$6a$; in all cases the general bound above gives $\Phi_\text{before}\leq M^2$. By \cref{lemma:phase-5-reduction,lemma:phase-6b-reduction}, each phase reduces the potential to at most a $(1- \frac{1}{4h^3})$ fraction of its previous value. Therefore,
\[
    \Phi_\text{after} \leq \left(1 - \frac{1}{4h^3}\right)^R \Phi_\text{before} \leq \left(1 - \frac{1}{4h^3}\right)^R M^2.
\]
If $\Phi_\text{after} = 0$, the outer loop terminates. Otherwise, \cref{lemma:Phi-lower-bound} gives
\[
2^{-O((m+n)^2(L + \log(m+n)))} \leq \left(1 - \frac{1}{4h^3}\right)^R M^2.
\]
Since $\log M = O(L + \log(m+n))$, taking logarithms and using $\ln(1-x) \leq -x$ gives
\[
    R = O(h^3(m+n)^2(L + \log(m+n))) = O((m+n)^5(L + \log(m+n))).
\]
By \cref{obs:finite-phases}, there are at most $m + K$ reset phases, i.e., phases of Type~$1$ or $6a$. Removing these reset phases partitions the remaining nonterminal phases of Type~$5$ and $6b$ into at most $m+K+1$ consecutive blocks, and the bound above applies to each block. There may also be one terminal phase, which does not affect the asymptotic bound. Hence the total number of phases is
\[
    O\!\left( (m+K)(m+n)^5(L + \log(m+n))\right).
\]
\end{proof}

\subsection{Work Per Phase}
\textbf{Step 4: work inside a phase.}
It remains to bound the number of inner iterations per phase and the number of max flow computations required per inner iteration.

\begin{proposition}
\label{prop:phase-time-bound}
Each phase consists of at most $O(mn)$ inner iterations, and each inner iteration performs at most $O(n + \log(m+n))$ max flow computations.
\end{proposition}
\begin{proof}
At the start of each phase, the computation time is dominated by computing a balanced flow of the expenditure network. At the start of each inner iteration, computing the scaling factor $\theta$ is dominated by computing the factor $\theta_5$ (for Event~5). As discussed in the algorithm overview, this requires at most $n$ max flow computations. In addition, an inner iteration may trigger Event~2, which requires recomputing a balanced flow, or Event~6, which requires one max flow/min-cut computation in the restricted network $N_{I,J}^{(i,0)}$ to classify the step as removal or refund. Using the parametric method of \citet{darwish2016improved}, each balanced flow can be computed with $O(\log(m+n))$ max flow computations. Hence each inner iteration uses at most $O(n + \log(m+n))$ max flow computations.

Each phase consists of $O(mn)$ inner iterations. This follows because the only events after which the algorithm remains in the inner loop are Events 2, 3, and 4. By \cref{lemma:event2-restoration}, Event~2 strictly enlarges $J$, so it can occur at most $n$ times in a phase. Between two Event~2 updates, the set $Z$ only shrinks: Event~3 removes one buyer from $Z$, and Event~4 removes one buyer from the market entirely. Since $Z$ is recomputed after each Event~2 update, there can be $O(m)$ Events~3 and~4 after each of the $O(n)$ Event~2 updates. Therefore, the number of max flow computations per phase is $O(mn(n + \log(m+n)))$. After the outer loop terminates, the final lower-bounded-flow instance can be reduced to $O(1)$ ordinary max flow computations by the standard circulation reduction recalled in Appendix~\ref{app:network-flow-background}; see also \citet{AhujaNetworkFlows}. Thus it does not affect the asymptotic bound. This completes the proof. 
\end{proof}

\noindent
We are now ready to prove the polynomial running time bound.

\begin{proof}[Proof of \cref{thm:running-time}]
By \cref{prop:phase-number-bound}, the algorithm executes at most
\[
O((m+K)(m+n)^5(L+\log(m+n)))
\]
outer-loop iterations. By \cref{prop:phase-time-bound}, each such phase contains $O(mn)$ inner iterations, and each inner iteration uses at most $O(n+\log(m+n))$ max-flow computations. The final lower-bounded-flow computation uses only $O(1)$ additional max-flow computations by the standard reduction to circulation \citep{AhujaNetworkFlows}. Multiplying these bounds gives
\[
    O( (m+K)mn(n + \log(m+n)) (m+n)^5(L + \log(m+n)) )
\]
max-flow computations in total. The outer loop therefore terminates within the claimed number of iterations, and \cref{prop:correctness} shows that the outcome returned after final processing is a competitive equilibrium.
\end{proof}

\section{Discussion}
Our algorithm shows that Arctic auctions can accommodate sophisticated separable seller costs without losing polynomial-time computability. This removes a significant computational barrier for applications such as central-bank liquidity provision and sovereign debt restructuring, where institutions must compare many feasible supply schedules before choosing an outcome. More broadly, it shows that seller-side objectives can be incorporated into a Fisher-style auction framework without abandoning the exact computation of equilibrium.

A natural next question is whether the running time can be made strongly polynomial, as was recently achieved for the Arctic auction without costs \citep{garg2026stronglypolynomialalgorithmarctic}.

The present model keeps seller costs separable across goods. Natural extensions include non-separable constraints that capture cross-instrument dependencies, such as aggregate maturity targets, currency exposure, or portfolio risk limits. Whether rationality of equilibrium prices and polynomial-time computability survive in such settings remains open, and may require new ideas beyond the hybrid min-cut invariant used here.

\bibliographystyle{abbrvnat}
\bibliography{refs}

\appendix

\section{Network Flow Background}
\label{app:network-flow-background}

We briefly recall the network-flow terminology used in the paper. For a comprehensive treatment, see \citet{AhujaNetworkFlows}.

A directed network is a directed graph $N=(V,A)$ with a nonnegative capacity $u(a)$ on every arc $a\in A$. We allow $u(a)=\infty$ as a notational convenience; in the expenditure networks of this paper, the infinite-capacity arcs are the goods-to-buyers arcs, and every finite cut avoids cutting such arcs. For a set of vertices $X\subseteq V$, let
\[
    \delta^+(X)=\{(v,w)\in A\mid v\in X,\ w\notin X\},
    \qquad
    \delta^-(X)=\{(v,w)\in A\mid v\notin X,\ w\in X\}.
\]
For any function $q$ on arcs, we write $q(\delta^+(X))=\sum_{a\in\delta^+(X)}q(a)$ and similarly for $\delta^-(X)$.

\paragraph{Flows and residual graphs.}
An $s$-$t$ flow is a function $f:A\to\R_{\geq 0}$ such that $0\leq f(a)\leq u(a)$ for every arc $a\in A$, and flow is conserved at every vertex $v\in V\setminus\{s,t\}$:
\[
    f(\delta^-(\{v\}))=f(\delta^+(\{v\})).
\]
The value of the flow is the net amount sent from $s$ to $t$. In the networks used in this paper there are no arcs entering $s$ or leaving $t$, so this value is simply $f(\delta^+(\{s\}))=f(\delta^-(\{t\}))$. A max flow is a feasible $s$-$t$ flow of maximum value.

Given a feasible flow $f$, the full residual graph is the directed graph on the same vertex set $V$ with a forward residual arc $(v,w)$ whenever $(v,w)\in A$ and $f(v,w)<u(v,w)$, and a reverse residual arc $(w,v)$ whenever $(v,w)\in A$ and $f(v,w)>0$. The residual capacity of a forward arc is $u(v,w)-f(v,w)$, while the residual capacity of a reverse arc is $f(v,w)$. A directed path from $s$ to $t$ in the full residual graph is an augmenting path. A feasible flow is maximum if and only if its full residual graph contains no augmenting path. In the algorithmic analysis, $R_N(f)$ denotes the restricted residual graph obtained from this full residual graph by deleting $s$, $t$, and all incident arcs.

\paragraph{Cuts and min-cuts.}
An $s$-$t$ cut is a partition $(X,V\setminus X)$ with $s\in X$ and $t\notin X$. Its capacity is $u(\delta^+(X))$. The max-flow min-cut theorem states that the maximum value of an $s$-$t$ flow is equal to the minimum capacity of an $s$-$t$ cut. In the expenditure network $N(\pb,\mb,\tb)$, if a set of goods $S$ is placed on the source side of a finite cut, then all adjacent buyers $\Gamma(S)$ must also be on the source side; otherwise the cut would include an infinite-capacity goods-to-buyers arc. For fixed $S$, the minimum such finite cut places exactly the buyers $\Gamma(S)$ on the source side. This is the observation behind \cref{lemma:source-cut-characterisation}.

When we refer to the source-side cut of an expenditure network, we mean the cut whose source side is just $\{s\}$. More generally, a source-side min-cut with good side $S$ and buyer side $T$ has source side $\{s\}\cup S\cup T$. The source-side cut $(\{s\},G\cup B\cup\{t\})$ is a min-cut exactly when moving any set of goods, together with its neighboring buyers, to the source side cannot decrease the cut capacity.

\paragraph{Lower-bounded flows and circulations.}
A lower-bounded network has lower and upper bounds $\ell(a)\leq u(a)$ on every arc. A feasible lower-bounded $s$-$t$ flow satisfies $\ell(a)\leq f(a)\leq u(a)$ on every arc and conserves flow at all vertices other than $s$ and $t$. A circulation is the analogous object with no distinguished source or sink: flow is conserved at every vertex.

The standard reduction from lower-bounded $s$-$t$ flow to circulation adds an auxiliary arc $(t,s)$, usually with infinite upper capacity and zero lower capacity. Any feasible $s$-$t$ flow becomes a feasible circulation after sending the same flow value on this auxiliary arc, and any feasible circulation restricts to a feasible $s$-$t$ flow after the auxiliary arc is removed.

We use Hoffman's circulation criterion in \cref{lemma:final-lb-feasible}. It states that a lower-bounded directed network admits a feasible circulation if and only if, for every vertex set $X\subseteq V$,
\[
    \ell(\delta^-(X))\leq u(\delta^+(X)).
\]
The necessity is immediate: in any circulation, the total flow entering $X$ equals the total flow leaving $X$, so the lower bounds on arcs entering $X$ cannot exceed the upper bounds on arcs leaving $X$. The sufficiency is a standard theorem; see \citet{AhujaNetworkFlows} for details.

\end{document}